\documentclass[11pt]{article} 
\usepackage[letterpaper,hmargin=1in,vmargin=1in]{geometry}

\def\ShowAuthNotes{1}
\usepackage[utf8]{inputenc}
\usepackage[T1]{fontenc}
\usepackage{lmodern}

\usepackage{float,url,hyperref,epic,eepic}
\usepackage{xcolor, enumerate}
\usepackage{graphicx}
\usepackage{amsmath, amsfonts, amssymb, amsthm, amscd}
\usepackage{calc}
\usepackage{algorithmic, algorithm}
\usepackage{listings}

\newcommand{\Pd}{\mathsf{P}}
\newcommand{\cA}{\mathcal{A}}\newcommand{\cB}{\mathcal{B}}

\newcommand{\cH}{\mathcal{H}}
\newcommand{\cI}{\mathcal{I}}\newcommand{\cJ}{\mathcal{J}}
\newcommand{\cK}{\mathcal{K}}\newcommand{\cL}{\mathcal{L}}
\newcommand{\cM}{\mathcal{M}}\newcommand{\cN}{\mathcal{N}}

\newcommand{\cS}{\mathcal{S}}

\newcommand{\cW}{\mathcal{W}}\newcommand{\cX}{\mathcal{X}}

\newcommand{\accept}{\text{accept}}
\newcommand{\reject}{\text{reject}}

\makeatletter
\newtheorem*{rep@theorem}{\rep@title}
\newcommand{\newreptheorem}[2]{%
\newenvironment{rep#1}[1]{%
 \def\rep@title{#2 \ref{##1}}%
 \begin{rep@theorem}}%
 {\end{rep@theorem}}}
\makeatother

\newcounter{hours}
\newcounter{minutes}
\newcommand{\printtime}{ %
        \setcounter{hours}{\time/60} %
        \setcounter{minutes}{\time-\value{hours}*60} %
        \ifthenelse{\value{hours}<10}{0}{}\thehours:%
        \ifthenelse{\value{minutes}<10}{0}{}\theminutes%
        } 

\theoremstyle{plain}
        \newtheorem{theorem}{Theorem}[section]
        \newreptheorem{theorem}{Theorem}
        \newtheorem{lemma}[theorem]{Lemma}
        \newreptheorem{lemma}{Lemma}
        \newtheorem{corollary}[theorem]{Corollary}
        \newtheorem{claim}[theorem]{Claim}

\theoremstyle{definition}
        \newtheorem{definition}[theorem]{Definition}

\theoremstyle{remark}
        \newtheorem{remark}[theorem]{Remark}

\DeclareMathOperator*{\Exp}{E}
\DeclareMathOperator*{\Var}{\text{Var}}

\newcommand{\inu}{\ensuremath{\leftarrow}}
\DeclareMathOperator{\poly}{poly}
\newcommand{\eps}{\ensuremath{\varepsilon}}

\newcommand{\NP}{\mathsf{NP}}

\renewcommand{\P}{\mathsf{P}}
\newcommand{\Bad}{\mathsf{Bad}}

\newcommand{\PSPACE}{\mathsf{PSPACE}}
\newcommand{\IP}{\mathsf{IP}}

\newcommand{\AM}{\mathsf{AM}}

\definecolor{DSgray}{cmyk}{0,0,0,0.7}

\definecolor{DSred}{cmyk}{0,0.7,0,0.7}

\ifnum\ShowAuthNotes=1
\newcommand{\Authornote}[2]{\noindent{\small\textcolor{DSgray}{\sf{
\textcolor{red}{[#1: #2]\marginpar{\textcolor{red}{\fbox{\Large !}}}}}}}}
\newcommand{\Authormarginnote}[2]{\marginpar{\parbox{2.2cm}{\raggedright\tiny \textcolor{red}{#1: #2}}}}
\else
\newcommand{\Authornote}[2]{}
\newcommand{\Authormarginnote}[2]{}
\fi

\begin{document}
\title{A Protocol for Generating Random Elements with
their Probabilities}

\author{
Thomas Holenstein\thanks{ETH Zurich, Department of Computer Science, 8092 Zurich, Switzerland. E-mail: {\tt thomas.holenstein@inf.ethz.ch}} \and
Robin K\"unzler\thanks{ETH Zurich, Department of Computer Science, 8092 Zurich, Switzerland. E-mail: {\tt robink@inf.ethz.ch}}}
\maketitle
\begin{abstract}
We give an $\AM$ protocol that allows the verifier to sample 
elements $x$ from a probability distribution $\Pd$, which is held by the prover. 
If the prover is honest, the verifier outputs $(x,\Pd(x))$ with probability close to $\Pd(x)$. 

In case the prover is dishonest, one may hope for the following guarantee: if the verifier outputs $(x,p)$, then the probability that the verifier outputs $x$ is close to $p$. Simple examples show that this cannot be achieved. Instead, 
we show that the following weaker condition holds (in a well defined sense) \textit{on average}: If $(x,p)$ is output, then $p$ is an upper bound on the probability that $x$ is output. 
	
Our protocol yields a new transformation to turn interactive proofs where the verifier uses private random coins into proofs with public coins. The verifier has better running time compared to the well-known Goldwasser-Sipser transformation (STOC, 1986). For constant-round protocols, we only lose an arbitrarily small constant in soundness and completeness, while our public-coin verifier calls the private-coin verifier only once.  
\end{abstract}

\newcommand{\Isize}{\text{IntervalSize}}
\newcommand{\Gsize}{\text{GapSize}}
\newcommand{\SampGap}{\text{SamplingGap}}

\newpage
\tableofcontents

\newpage
\section{Introduction}
\label{chap:Sampling}

In an interactive proof \cite{GMR85, Bab85, BM88}, an all-powerful prover tries to convince a computationally bounded verifier that some statement is true. The study of such proofs has a rich history, and has lead to numerous important and surprising results. 

We are interested in interactive protocols that allow the verifier to sample elements from a probability distribution. Such protocols have proved to be very useful, and their applications include the study of private versus public coins in interactive proof systems \cite{GS86}, perfect zero knowledge \cite{For87}, 
basing average-case hardness or cryptographic security on worst-case hardness \cite{FF93, BT06, AGGM06, HMX10}, and many more.

We consider constant-round protocols that allow the verifier to sample an element $x$ from a probability distribution $\Pd$ together with an approximation $p$ of the probability $\Pd(x)$.
The verifier outputs pairs $(x,p)$, and for a fixed prover we let $(X,P)$ be the random variables corresponding to the verifier's output, $\Pd_{XP}$ is their joint distribution, and $\Pd_X$ is the marginal distribution defined as $\Pd_X(x)= \sum_p\Pd_{XP}(x,p)$. We would like to achieve the following properties:

\begin{quote}
	\textbf{Property~1:} 
		For every $x$ we have $\Pd_X(x) \approx \Pd(x)$. 
\end{quote}

\begin{quote}
	\textbf{Property~2:} 
		If the verifier outputs $(x,p)$, then 
			$p \approx \Pd_X(x)$. 
\end{quote}

Recently, Haitner et al.~\cite{HMX10} gave such a protocol for sampling a distribution $\Pd$ on bit strings, which is given as $\Pd=f(\Pd_U)$ for an efficiently computable function $f$, where $\Pd_U$ is the uniform distribution on $n$-bit strings. For \textit{any} (possibly cheating) prover their protocol achieves property $1$ with equality (i.e.~$\Pd_X(x)=\Pd(x)$), and property $2$ for polynomially small error (i.e.~$p = (1\pm \eps)\Pd(x)$ for $\eps$ polynomially small in $n$). The protocol extends to distributions $f(\Pd_S)$, where $\Pd_S$ is the uniform distribution on an efficiently decidable set $\cS \subseteq \{0,1\}^n$. 

\subsection{Contributions of this Paper}
We give a sampling protocol that is similar to the one of~\cite{HMX10}. However, in our protocol only the prover gets as input the distribution $\Pd$. This distribution can be arbitrary, and in particular does not have to be efficiently samplable. The verifier does \textit{not} get $\Pd$ as input, and in particular does not have the ability to sample from $\Pd$. 

We obtain the following completeness guarantee: 
\begin{quote}
	\textbf{Completeness:} If the prover is honest, then both 
		properties 1 and 2 are satisfied. 
		
		More precisely, property~1 is satisfied with polynomially
		small error (i.e.~$\Pd_X(x) = (1\pm\eps)\Pd(x)$ for 
		polynomially small $\eps$), and instead of property 2 we
		even guarantee that the verifier only outputs pairs 
		$(x,\Pd(x))$. 
\end{quote}
In case the prover is not honest, since the verifier does not know $\Pd$, we cannot hope to satisfy property~1. However, one could hope that property~2 is satisfied for any (possibly cheating) prover. Unfortunately, simple examples show that this cannot be achieved. Instead, we prove the following weaker guarantee:
\begin{quote}
	\textbf{Soundness:} The following condition holds (in a well defined sense) \textit{on average}: If $(x,p)$ is output, then $\Pd_X(x) \leq p$. 
\end{quote}

To illustrate the usefulness of our protocol, we apply it to obtain a private-coin to public-coin transformation for interactive proofs in Section~\ref{chap:PrivCoinsPubCoins}. Compared to the original transformation by Goldwasser and Sipser~\cite{GS86}, our transformation is more efficient in terms of the verifier's running time. In particular, for constant-round protocols we only lose an arbitrarily small constant in soundness and completeness, while executing the private-coin verifier exactly once.
We show that this transformation can be viewed as an interactive sampling process, where the public-coin verifier iteratively samples messages for the private-coin verifier and requests the corresponding answers from the prover. 

We remark that the soundness guarantee of our protocol is weaker than the one of the~\cite{HMX10} protocol. However, as mentioned above, the guarantee of~\cite{HMX10} cannot be achieved in the more general setting we consider. Furthermore, to achieve a private-coin to public-coin transformation, it is also possible to employ the~\cite{HMX10} sampling protocol. However, this yields a less efficient public-coin verifier, as the private coin-verifier must be executed many times when running the sampling protocol. 

\subsection{Related Work}

\paragraph{Interactive protocols.}
Interactive proof systems were introduced by Goldwasser et al.~\cite{GMR85}. Independently, Babai and Moran \cite{Bab85, BM88} defined the public-coin version. As mentioned above, Goldwasser and Sipser~\cite{GS86} showed that the two definitions are equivalent with respect to language recognition (for a nice exposition of the proof, we also refer to the book of Goldreich~\cite{Gol08}). The study of interactive proofs has a long and rich history, influencing the study of zero knowledge and probabilistic checkable proofs, and has lead to numerous important and surprising results such as $\IP = \PSPACE$~\cite{LFKN92, Shamir92}. For historical overviews we refer for example to \cite{Bab90, Gol08, AB09}. 

\paragraph{Interactive sampling protocols.}
Goldwasser and Sipser \cite{GS86} show that private coins in interactive proofs can be made public. A constant-round set lower bound protocol is introduced which can be viewed as a sampling process: the verifier uses pairwise independent hashing to randomly select a few elements in a large set. This protocol is used in many subsequent works such as \cite{For87,AH91,GVW02,BT06,AGGM06,HMX10}. 
To study the complexity of perfect zero knowledge, Fortnow \cite{For87} introduces a constant-round protocol that allows to prove set upper bounds assuming that the verifier is given a uniform random element in the set, which is not known to the prover. The same protocol is used in a similar context by Aiello and H{\aa}stad \cite{AH91}. As in the lower bound protocol, hashing is used to sample a small number of elements from the set.
To prove that any interactive proof system can have perfect completeness, Goldreich et al.~\cite{GMS87} give a protocol that allows to sample a perfectly uniform random element from a decidable set. Their protocol requires a polynomial number of rounds (depending on the set size), and they show that no constant-round protocol can achieve this task. 
The upper and lower bound protocols of \cite{GS86, For87, AH91} are used by Bogdanov and Trevisan \cite{BT06} in their proof that the worst-case hardness of an $\NP$-complete problem cannot be used to show the average-case hardness of an $\NP$ problem via non-adaptive reductions, unless the polynomial hierarchy collapses.
The ideas used in the set lower and upper bound protocols can be employed to sample a single element from an $\NP$ set in case the verifier knows the set size. This is done in \cite{GVW02} in the context of studying interactive proofs with bounded communication.
This protocol is refined by Akavia et al.~\cite{AGGM06}, where it is used to give protocols for proving the size of a set (both upper and lower bound) under the assumption that the verifier knows some approximate statistics about the size of the set. These protocols are then used to study the question whether one-way functions can be based on $\NP$-hardness. 
The ideas behing the sampling protocol by Akavia et al.~are refined and extended by Haitner et al.~\cite{HMX10} in order to give a protocol that allows to sample an element $x \leftarrow \Pd$ of a given distribution, which is specified as $\Pd = f(\Pd_U)$ for an efficiently computable function $f$, where $\Pd_U$ is the uniform distribution on $\{0,1\}^n$. The verifier outputs $(x,p)$, 
where $x$ is sampled from $\Pd$, and $p = (1\pm\eps)\Pd(x)$ for polynomially small $\eps$. Their protocol extends to distributions 
$f(\Pd_S)$, where $\Pd_S$ is the uniform distribution on an efficiently decidable set $\cS \subseteq \{0,1\}^n$. It is shown that this sampling protocol can be employed to allow the verifier to verify the shape of the distribution $\Pd$, in terms of its histogram. These protocols are used to prove that a number of cryptographic primitives, such as statistically hiding commitment, cannot be based on $\NP$-hardness via certain classes of reductions (unless the polynomial hierarchy collapses).

\section{Preliminaries}
\label{sec:SamplingPreliminaries}

\subsection{Notation and Definitions}
We denote sets using calligraphic letters $\cA, \cB, \ldots$. We usually denote the elements of a set $\cX$ using lower case letters $x_1, x_2, \ldots$. For $n\in \mathbb{N}$ we let $(n):=\{0, 1, \ldots, n\}$ and $[n]:=\{1,2,\ldots, n\}$.

For a finite set $\cX$, a probability distribution $\Pd_X$ over $\cX$ is a function $\Pd_X: \cX \rightarrow [0,1]$ that satisfies the condition $\sum_{x\in \cX}\Pd_X(x)=1$. For a set $\cM \subseteq \cX$ we let $\Pd_X(\cM)=\sum_{x\in \cM}\Pd_X(x)$. 
 We will denote random variables by capital letters (e.g.~$X$), and values they assume by lower case letters (e.g.~$x$). Probability distributions that are not specifically associated with a random variable will be denoted by $\Pd$.

We say a function $f:\mathbb{N}\rightarrow\mathbb{N}$ is \textit{time-constructible} if there exists an algorithm that runs in time $O(f(n))$ and on input $1^n$ outputs $f(n)$.

\subsection{Concentration Bounds}
\label{sec:ConcBounds}
We use several concentration bounds, and first state the well-known Chebyshev and Chernoff bounds. 
\begin{lemma}[Chebyshev's inequality]
	\label{lem:Chebyshev}
	For any random variable $X$ where $\Exp[X]$ and $\Var[X]$ are finite, and any $k \in \mathbb{R}_{>0}$ we have
		$$\Pr\left[\big|X-\Exp[X]\big|\geq k\cdot \sqrt{\Var[X]}\right]\leq 1/k^2.$$
\end{lemma}
\begin{lemma}[Chernoff bound]
	\label{lem:ChernoffBound}
	Let $X_1, \ldots, X_k$ be independent random variables where for all $i$ we have $X_i \in \{0,1\}$ and $\Pr[X_i=1]=p$ for some $p \in (0,1)$. Define $\widetilde{X}:=
	\frac{1}{k}\sum_{i\in [k]} X_i$. 	
	Then for any $\eps >0$ it holds that
	\begin{align*}
		&\Pr_{X_1, \ldots, X_k}\left[\widetilde{X} \geq p+\eps\right]
			<\exp{\left( 
					-\frac{\eps^2k}{2}
				\right)}, \quad
		\Pr_{X_1, \ldots, X_k}\left[\widetilde{X} \leq p-\eps\right]
			<\exp{\left( 
					-\frac{\eps^2k}{2}
				\right)}.
	\end{align*}
\end{lemma}

Hoeffding's bound \cite{Hoe63} states that for $k$ independent random variables $X_1, \ldots, X_k$ that take values in some appropriate range, with high probability their sum is close to its expectation. 
\begin{lemma}[Hoeffding's inequality]
	\label{lem:HoeffdingBound}
		Let $X_1, \ldots, X_k$ be independent random variables
	with $X_i \in [a,b]$, define $\widetilde{X}:=
	\frac{1}{k}\sum_{i\in [k]} X_i$ and let 
	$p=\Exp_{X_1, \ldots, X_k}[\widetilde{X}]$. 
	Then for any $\eps > 0$ we have
	\begin{align*}
		&\Pr_{X_1, \ldots, X_k}\left[
		\widetilde{X} \geq p
			+ \eps			
		\right]
		\leq \exp\left(-\frac{2\eps^2k}	
																{(b-a)^2}\right), \\
		&\Pr_{X_1, \ldots, X_k}\left[
		\widetilde{X} \leq
			p- \eps			
		\right]
		\leq \exp\left(-\frac{2\eps^2k}	
																{(b-a)^2}\right).
		\end{align*}
\end{lemma}

\subsection{Interactive Proofs}
\label{sec:InteractiveProofs}

In an interactive proof, an all-powerful prover tries to convince a computationally bounded verifier that her claim is true. The notion of an interactive protocol formalizes the interaction between the prover and the verifier, and is defined as follows. 

\begin{definition}[Interactive Protocol]
\label{def:InteractiveProtocol}
Let $n \in \mathbb{N}$,
$V:\{0,1\}^{\ast} \times \{0,1\}^{\ast} \rightarrow \{0,1\}^{\ast} \cup \{\accept, \reject\}$, 
$P:\{0,1\}^{\ast} \rightarrow \{0,1\}^{\ast}$, and $k, \ell, m: \mathbb{N}\rightarrow \mathbb{N}$. A 
$k$-round interactive
protocol $(V,P)$ with message length $m$ and $\ell$ random coins between $V$ and $P$ on input $x\in \{0,1\}^n$ is defined as follows:
\begin{enumerate}
	\item Uniformly choose random coins $r \in \{0,1\}^{\ell(n)}$ 
				for $V$.
	\item Let $k:=k(n)$ and repeat the following for $i=0,1,\ldots,k-1$:
				\begin{enumerate}[a)]	
					\item $m_i := V(x,i,r,a_0, \ldots, a_{i-1}), m_i \in \{0,1\}^{m(n)}$
					\item $a_i := P(x,i,m_0, \ldots, m_i), a_i \in \{0,1\}^{m(n)}$
				\end{enumerate}
				Finally, we have $V(x,k,r,m_0,a_0, \ldots, m_{k-1},
				a_{k-1}) \in \{\accept, \reject\}$.
\end{enumerate}
We denote by $(V(r),P)(x) \in \{\accept, \reject\}$ the output of $V$ on random coins $r$ after an interaction with $P$.
We say that $(x,r,m_0,a_0, \ldots, m_{j},a_{j})$ is consistent for $V$ if for all $i \in (k-1)$
we have $V(x,i,r,a_0, \ldots, a_{i-1})=m_i$. Finally, 
if $(x,r,m_0,a_0, \ldots, m_{k-1},	a_{k-1})$ is not consistent
for $V$, then $V(x,k,r,m_0,a_0, \ldots, m_{k-1}, 
a_{k-1}) = \reject$. 
\end{definition}
We remark that we restrict the verifier to only accept consistent transcripts as this will simplify the notation in Section~\ref{chap:PrivCoinsPubCoins}.

We now define the classes $\IP$ and $\AM$ of interactive proofs. The definition of $\IP$ was initially given by Goldwasser, Micali, and Rackoff~\cite{GMR85}, and the definition of $\AM$ goes back to Babai~\cite{Bab85}. 

\begin{definition}[$\IP$ and $\AM$]	
	\label{def:InteractiveProofs}
	The set
	\begin{align*}
		\IP\left(
			\begin{array}{lll}
					\textsf{rounds}& =    & k(n)	 \\ 
					\textsf{time}  & =    & t(n) 	 \\
					\textsf{msg size}  & =    & m(n) 	 \\
					\textsf{coins} & =    & \ell(n)\\
					\textsf{compl} & \geq & c(n)   \\ 
					\textsf{sound} & \leq & s(n)   \\ 
			\end{array}
		\right)
	\end{align*}
	contains the languages $L$ that admit a $k$-round
	interactive protocol $(V,P)$ with message length $m$ and 
	$\ell$ random coins, and the following properties:
	\begin{itemize}
		\item[]
	\noindent \textbf{Efficiency:} 
	$V$ can be computed by an algorithm such that
	for any $x \in \{0,1\}^*$ and $P^*$ the total running time
	of $V$ in $(V,P^*)(x)$ is at most $t(|x|)$.
	
	\noindent \textbf{Completeness: }
	$$
		x \in L \implies \Pr_{r \leftarrow \{0,1\}^{\ell(|x|)}}
					\left[ (V(r),P)(x) = \accept
					\right] 
					\geq c(|x|).
	$$
	
	\noindent \textbf{Soundness: }
	For any $P^*$ we have
	$$
		x \notin L \implies \Pr_{r \leftarrow \{0,1\}^{\ell(|x|)}}
					\left[ (V(r),P^*)(x) = \accept
					\right] 
					\leq s(|x|).
	$$
	\end{itemize}
	The set $\AM$ is defined analogously, with the additional
	restriction that $(V,P)$ is \textit{public-coin}, i.e.~for
	all $i$, $m_i$ is an independent uniform random string. 
	We sometimes omit the \textsf{msg size} and
	\textsf{coins} parameters from the notation, in which case
	they are defined to be at most \textsf{time}. If we omit the \textsf{time} parameter, it is defined to be $\poly(n)$.  We then let
	\begin{align*}
		\IP:=
		\IP\left(
			\begin{array}{lll}
					\textsf{rounds}& =    & \poly(n)	 \\ 
					\textsf{compl} & \geq & 2/3   \\ 
					\textsf{sound} & \leq & 1/3   \\ 
			\end{array}
		\right),		\quad
		\AM:=
		\AM\left(
			\begin{array}{lll}
					\textsf{rounds}& =    & 1	 \\ 
					\textsf{compl} & \geq & 2/3   \\ 
					\textsf{sound} & \leq & 1/3   \\ 
			\end{array}
		\right)
	\end{align*}
\end{definition}

Instead of writing (for example) $L \in \AM(\textsf{rounds}=k,\textsf{time}=t,\textsf{compl}\geq c, \textsf{sound}\leq s)$, we sometimes say that $L$ has a $k$-round public-coin interactive proof with completeness $c$ and soundness $s$, where the verifier runs in time $t$. 

Babai and Moran~\cite{BM88} showed that in the definition of $\AM$ above, setting $\textsf{rounds}=k$ for any constant $k\geq 1$ yields the same class. 

\paragraph{Assuming deterministic provers.} 
For proving the soundness condition of an interactive proof, without loss of generalty we may assume that the prover is determinsitic: we consider the deterministic prover that always sends the answer which maximizes the verifier's acceptance probability. No probabilistic prover can achieve better acceptance probability.  

\subsection{Hash Functions}
\label{sec:HashFunctions}
A family of pairwise independent hash functions satisfies the property that for any $x,x'$, the distribution $(h(x),h(x'))$ over the choice of $h$ from the family is uniform. More generally, we can define this concept for $k$-tuples:

\begin{definition}[$k$-wise independent hash functions]
	A family of functions $\cH(n,m) = \{h: \{0,1\}^n 
		\rightarrow \{0,1\}^m\}$ is $k$-wise independent
		if for any $x_1, \ldots, x_k \in \{0,1\}^n$ and
		$y_1, \ldots, y_k \in \{0,1\}^m$ we have
		$\Pr_{h \inu \cH(n,m)}[\forall i\in [k]: h(x_i)=y_i] 
		 = 2^{-km}$.
\end{definition}
Such families exist, and can be efficiently sampled for any 
$k=\poly(n,m)$ \cite{CW79}. 
The following lemma, as given by Nisan \cite{Nis92}, states that for any set $\cB \subseteq \{0,1\}^n$, 
if we choose a pairwise independent hash function $h$ from $n$ to $m$ bits, then the size of the 
set $\{y \in \cB: h(y)=0^m\}$ is close to its expected value $\frac{|\cB|}{2^m}$ 
with high probability. If a $3$-wise independent hash function is used, this property holds
even conditioned on $h(x)=0$ for some fixed $x$. 
\begin{lemma}[Hash Mixing Lemma]
		\label{lemma:hashMixingLemma}
		Let $\cB \subseteq \{0,1\}^n$, $x \in \{0,1\}^n$.
		If $\cH(n,m)$ is a family of $2$-wise independent hash functions mapping $n$ bits to $m$ bits, then the following holds.
		\begin{enumerate}[(i)]
			\item For all $\gamma > 0$ we have
				\begin{align*}
				   \Pr_{h\inu \cH(n,m)}
					 &\left[ |\{y\in \cB: h(y)=0^m\}|\notin
					 			 (1\pm\gamma)\frac{|\cB|}{2^m}
					 \right]\leq 
					 \begin{cases} 
					 		 \frac{2^m}{\gamma^2|\cB|}
							    &\mbox{if } |\cB|>0 \\
							0   &\mbox{if } |\cB|=0.
					 \end{cases}
				\end{align*}
		\end{enumerate}
		If $\cH(n,m)$ is a family of $3$-wise independent hash functions mapping $n$ bits to $m$ bits, then the following holds.
		\begin{enumerate}
			\item[(ii)] If $x \notin \cB$, then for all $\gamma > 0$ we have
				\begin{align*}
				   \Pr_{h\inu \cH(n,m)}
					 &\left[|\{y\in \cB: h(y)=0^m\}|\notin
					 			 (1\pm\gamma)\frac{|\cB|}{2^m}
					 \Big|h(x)=0^m\right]  \leq 
					 \begin{cases} 
					 		 \frac{2^m}{\gamma^2|\cB|}
							    &\mbox{if } |\cB|>0 \\
							0   &\mbox{if } |\cB|=0.
					 \end{cases}
				\end{align*}
			\item[(iii)] If $x \in \cB$, then for all $\gamma > 0$
						we have
				\begin{align*}
				 	 \Pr_{h\inu \cH(n,m)}
					 &\left[|\{y\in \cB: h(y)=0^m\}|\notin
					 			 1+(1\pm\gamma)\frac{|\cB|-1}{2^m}
					 \Big|h(x)=0^m\right] \leq
					 \begin{cases} 
					 		\frac{2^m}{\gamma^2(|\cB|-1)}
							    &\mbox{if } |\cB|>1 \\
							0   &\mbox{if } |\cB|=1.
					 \end{cases}
				\end{align*}
		\end{enumerate}
\end{lemma}

\begin{proof}[Proof of Lemma~\ref{lemma:hashMixingLemma}]
	We first prove (i). 
	The case $|\cB|=0$ is trivial, so assume $|\cB|>0$. 
	The proof uses Chebyshev's inequality (see 
	Lemma~\ref{lem:Chebyshev}).
	In the following, all probabilities, expectations and 
	variances are	over $h\inu \cH(n,m)$. For some event $\mathsf{A}$ we
	let $[\mathsf{A}]=1$	if $\mathsf{A}$ occurs, and $[\mathsf{A}]=0$ otherwise.
	We define the random variable 
	$Z:=|\{y \in \cB: h(y) = 0^m\}|$ and find that
	\begin{align*}
		\Exp\left[
			Z
		\right]
		=\Exp\left[
			\sum_{y\in \cB}[h(y)=0]
		\right]
		=\sum_{y\in \cB}\Exp\left[
			[h(y)=0]
		\right]	
		= \sum_{y\in \cB}\Pr[h(y)=0]
		= \frac{|\cB|}{2^m},
	\end{align*}
	and
	\begin{align*}
		&\Exp\left[
			Z^2
		\right]
		=\Exp\left[
			|\{y \in \cB: h(y) = 0^m\}|^2
		\right]
		=\Exp\left[
			|\{(y,y') \in \cB^2: h(y)=h(y')=0^m\}|
		\right]\\
		&=\Exp\left[
			\sum_{(y,y')\in \cB^2}
			[h(y)=h(y')=0^m]
		\right]
		=\sum_{(y,y')\in \cB^2}\Exp\left[
			[h(y)=h(y')=0^m]
		\right]\\
		&=\sum_{(y,y')\in \cB^2}\Pr\left[
			h(y)=h(y')=0^m
		\right]
		=|\cB|^2\Pr\left[
			h(y)=h(y')=0^m
		\right]\\
		&=|\cB|^2
		\Bigl(
		\underbrace{\Pr\left[
			h(y)=h(y')=0^m | y = y'
				\right]}_{=\frac{|\cB|}{2^m} `}
			\cdot \frac{1}{|\cB|^2}\\
		&\qquad\qquad +
		\underbrace{\Pr\left[
			h(y)=h(y')=0^m | y \neq y'
				\right]}_{=\frac{1}{2^{2m}}} 
			\cdot \frac{|\cB|^2-|\cB|}{|\cB|^2}
		\Bigr) \\
		&= \frac{|\cB|}{2^m} +\frac{|\cB|^2-|\cB|}{2^{2m}}.
	\end{align*}
	In the fourth equality above we used the pairwise
	independence of $\cH(n,m)$. Thus we have
	\begin{align*}	
		\Var[Z]=\Exp[Z^2]-\Exp[Z]^2 = \frac{|\cB|}{2^m} - 
																	\frac{|\cB|}{2^{2m}}
	\end{align*}
	Now applying Chebyshev's inequality 
	(Lemma~\ref{lem:Chebyshev}) for $k=\frac{\gamma\Exp[Z]}{\sqrt{\Var[Z]}}$ gives the 
	claim, as
	\begin{align*}
		\Pr\left[
			|Z-\Exp[Z]|\geq \gamma \Exp[Z]
		\right]
		\leq
		\frac{1}{k^2} = \frac{2^m}{\gamma^2|\cB|}- \frac{1}{|\cB
		|\gamma^2} \leq \frac{2^m}{\gamma^2|\cB|}.
	\end{align*}
	
	Finally, (ii) and (iii) follow by the following argument:
	As we have $3$-wise independence, we still have pairwise
	independence conditioned on $h(x)=0^m$, and we can apply 
	(i). 
\end{proof}

\subsection{Histograms}
\label{sec:HistAndWasserstein}
We give the definition of histograms as given in \cite{HMX10}. The histogram of a probability distribution $\Pd$ is a function $h:[0,1] \rightarrow [0,1]$ such that $h(p)=\Pr_{x\leftarrow \Pd}[\Pd(x)=p]$. The following definition describes a discretized version of this concept.\begin{definition}[$(\eps, t)$-histogram]
	\label{def:histogram}
	Let $\Pd$ be a probability distribution on $\{0,1\}^n$, fix
	$t\in \mathbb{N}$, and let $\eps > 0$. 
	For $i \in (t)$ we define the $i$'th interval $\cA_i$ and the
	$i$'th bucket $\cB_i$ as
	\begin{align*}
			\cA_i := \left(2^{-(i+1)\eps}, 
								2^{-i\eps}\right], \qquad
			\cB_i := \left\{x: \Pd(x) \in \cA_i\right\}.
	\end{align*}
	We then let $h:=(h_0, \ldots, h_t)$ where 
	$h_i := \Pr_{x\leftarrow \Pd}[x \in \cB_i] = 
					\sum_{x\in \cB_i}\Pd(x)$. The tuple
	$h$ is called the $(\eps,t)$-histogram of $\Pd$.
\end{definition}
If for all $x$ we have either $\Pd(x)=0$ or 
$\Pd(x)\geq 2^{-n}$, and we consider the $(\eps,t)$-histogram of $\Pd$ for $t=\lceil n/\eps\rceil$, then $\bigcup_{i\in (t)}\cB_i = \{0,1\}^n$ and $\sum_{i\in (t)}h_i = 1$. If smaller probabilities occur (e.g.~$\Pd(x)=2^{-2n}$ for some $x$), this sum is smaller than $1$. 

The following observation follows directly from the above definition:
\begin{claim}
	\label{claim:sizeBi}
	For all $i\in (t)$ we have $h_i 2^{i\eps} 
					\leq |\cB_i| \leq h_i 2^{(i+1)\eps}$.
\end{claim}

\begin{proof}
	Recall that $h_i = \sum_{x\in \cB_i}\Pd(x)$ and
	by definition of $\cB_i$ we have
	$$ 
	|\cB_i|2^{-(i+1)\eps}
		=
			\sum_{x\in \cB_i}2^{-(i+1)\eps} 
		\leq
			\sum_{x\in \cB_i}\Pd(x) 
		\leq 
			\sum_{x\in \cB_i}2^{-i\eps} 
			= |\cB_i|2^{-i\eps}.\qedhere 
	$$
\end{proof}

\section{The Sampling Protocol}

\subsection{Informal Theorem Statement and Discussion}
\label{sec:SamplingDiscussion}
In our sampling protocol, the verifier will output pairs $(x,p) \in \{0,1\}^n \times (0,1]$. For a fixed prover, we let $(X,P)$ be the random variables corresponding to the verifier's output, and we denote their joint distribution by $\Pd_{XP}$. Also, $\Pd_X$ is defined by $\Pd_X(x)=\sum_p \Pd_{XP}(x,p)$. Informally, our sampling theorem can be stated as follows. 
\begin{theorem}[The Sampling Protocol, informal]
	\label{thm:SamplingInformal}
	There exists a constant-round public-coin interactive
	protocol\footnote{
		 See Definition~\ref{def:InteractiveProtocol} for a 
		 standard definition of interactive protocols. 
	} 
	such that the following holds.
	The verifier and the prover take as input 
	$n \in \mathbb{N}, \eps, \delta \in (0,1)$, 
	and the prover additionally gets as input a probability
	distribution $\Pd$ over $\{0,1\}^n$.
	The verifier runs in time
	$\poly\left(n \left(\frac{1}{\eps}\right)^{1/\delta}\right)$
	and we have:
	
	\vspace{0.2cm}	
	\noindent\textbf{Completeness: }
	If the prover is honest, then the verifier outputs $(x,\Pd(x))$
	with probability $(1\pm\eps)\Pd(x)$. 
	
	\vspace{0.2cm}	
	\noindent\textbf{Soundness: }
	Fix any (possibly cheating) prover. Then for all 
	$x\in \{0,1\}^n$ we have
	$\sum_{p} \frac{\Pd_{XP}(x,p)}{p}
	 \leq 1+\eps+\delta$.
\end{theorem}
The soundness condition may not be very intuitive at first sight. We therefore discuss it in detail below. 

To keep the discussion simple, the above theorem statement is only almost true: in fact, the completeness only holds with probability $1-\eps$ over the choice of $x$ from $\Pd$, and the soundness condition only holds if we condition on good protocol executions (where an execution is bad with probability at most $\eps$). We omit this here, and refer to Theorem~\ref{thm:main} for the exact statement.

We will prove soundness only for deterministic provers, but the same statement holds in case the prover is probabilistic: this follows easily by conditioning on the prover's random choices, and applying the result for the deterministic prover.\footnote{
	We remark that our definition of interactive protocols (Definition~\ref{def:InteractiveProtocol}) does not allow the prover to use randomness. 
	This is because when considering decision problems, the prover can always be assumed to be deterministic, as explained in Section~\ref{sec:InteractiveProofs}. 
	However, since our theorem does not consider a decision problem, and the sampling protocol might be used as a subprotocol of some other protocol, we do not assume that the prover is deterministic. 
}

In the sampling protocol of \cite{HMX10} (as described above) the verifier gets access to the distribution $\Pd$, in the sense that a circuit computing a function $f$ is provided as input to the protocol, where $\Pd = f(\Pd_U)$ and $\Pd_U$ is the uniform distribution. During the execution of the sampling protocol, the function $f$ needs to be evaluated many times. In contrast, in our protocol the verifier does not get access to the distribution, and never needs to evaluate such a circuit.

We note that the verifier runs in polynomial time for any polynomially small $\eps$ and constant $\delta$. It is an interesting open problem if it is possible to improve the protocol to allow both polynomially small $\delta$ and $\eps$ and an efficient verifier. 

\paragraph{Discussion of the soundness condition.}
The main motivation for the soundness condition is that it is actually sufficient for applying the sampling protocol to obtain the private-coin to public-coin transformation in Section~\ref{chap:PrivCoinsPubCoins}. 

However, there are several remarks we would like to discuss. We first give an overview of these remarks, and then discuss them in detail below. 
\begin{itemize}
	\item \textbf{Remark~1:} Fix any (possibly cheating) prover and recall property~2 as described in the introduction: If the verifier outputs $(x,p)$, then $p \approx \Pd_X(x)$. One may hope to give a protocol that satisfies this property. However, we show below that this cannot be achieved in our setting.
	\item \textbf{Remark~2:} It is possible to interpret the soundness condition as follows: The situation $\Pd_X(x) \gg p$ cannot occur too often. In that sense, the protocol provides an upper bound on $\Pd_X(x)$ ``on average''. 
	\item \textbf{Remark~3:} Assume the completeness condition is satisfied. Then the soundness condition actually holds if the prover behaves like a convex combination of honest provers, i.e.~it first chooses a distribution $\Pd$ from a set of distributions, and then behaves honestly for $\Pd$. 
	\item \textbf{Remark~4:} 
	It is natural to ask if the converse of the statement in remark~3 is also true. That is, we ask the following. Suppose we have some protocol that satisfies both our completeness and soundness conditions. Fix any prover and consider the verifier's output distribution $\Pd_{XP}$. 
	Is there a convex combination of probability distributions, such that if the prover first chooses $\Pd$ and then behaves honestly for $\Pd$, the verifier's output distribution equals $\Pd_{XP}$? 	
	We show below that this is not the case in general. 
\end{itemize}
Proving the soundness condition that any cheating prover can be seen as a convex combination of honest provers would imply that the protocol is optimal, since for any protocol a probabilistic prover actually \textit{can} first choose some distribution and then behave honestly for it. Remark~4 implies that our soundness condition does not imply this in general.

\paragraph{Remark~1: Property~2 cannot be achieved.}
 Unfortunately, property~2 is impossible to achieve, and we now describe a strategy for a dishonest (probabilistic) prover that shows this. Let $\Pd_0$ be the uniform distribution on $\{0^n, 1^n\}$ and $\Pd_1$ the distribution that outputs $0^n$ with probability $1$. 
The prover chooses $b\inu \{0,1\}$ uniformly at random, and then behaves like the honest prover for $\Pd_b$ (i.e.~conditioned on $b$, the completeness condition is satisfied for $\Pd_b$). We find that $\Pd_X(0^n) \approx 3/4$ and $\Pd_X(1^n) \approx 1/4$. 
But still we have for both $x\in \{0^n,1^n\}$ that
$\Pd_{XP}(x,1/2) \approx 1/4$. Thus, if the verifier outputs $(x,p)$, both variants $\Pd_X(x) \gg p$ and $\Pd_X(x) \ll p$
are possible.\footnote{
	To be very explicit, we have 
	$\Pd_{XP}(0^n,1/2) = 1/4$, and 
	$\Pd_{X}(0^n) =3/4 \gg 1/2$. On the other hand, 
	we find 
	$\Pd_{XP}(1^n,1/2) = 1/4$, and
	$\Pd_{X}(1^n) =1/4 \ll 1/2$. 
}

\paragraph{Remark~2: The situation $\Pd_X(x) \gg p$ cannot occur too often.}
We cannot hope to avoid $\Pd_X(x) \ll p$, since the prover can always choose to make the verifier reject with, say, probability $1/2$, and behave honestly for a given distribution $\Pd$ otherwise. In that case for any $x$ we have $\Pd_X(x) = \Pd_{XP}(x,\Pd(x)) \approx \Pd(x)/2 \ll \Pd(x)$. 

We now consider the case $\Pd_X(x) \gg p$. An interpretation of our soundness condition is that the situation $\Pd_X(x) \gg p$ cannot occur too often, and we now explain this.
For any fixed $x$ we find
\begin{align*}
	\sum_{p} \frac{\Pd_{XP}(x,p)}{p}
	= \sum_{p} \Pd_{P\mid X}(p,x) \frac{\Pd_X(x)}{p} 
	=\Exp_{p\leftarrow \Pd_{P\mid X=x}}
		\left[\frac{\Pd_X(x)}{p}\right]
\end{align*}
and thus the soundness condition is equivalent to 
\begin{align*}
	\Exp_{p\leftarrow \Pd_{P\mid X=x}}
		\left[\frac{\Pd_X(x)}{p}\right] \leq 1+\eps+\delta.
\end{align*}
Indeed, an interpretation of this statement is that for any $x$, on average, $P$ is not much smaller than $\Pd_X(x)$. In this sense, $p$ provides an upper bound on $\Pd_X(x)$ \textit{on average}. 

\paragraph{Remark~3: The soundness condition holds for a convex combination of provers.}
We proceed to argue that the soundness condition actually holds in case the prover first chooses a distribution $\Pd$ from a set of distributions and then behaves honestly for $\Pd$. First observe that for the prover strategy using $\Pd_0$ and $\Pd_1$ as described above (see remark~1), our soundness property actually holds for $0^n$: 
$$ 
\sum_{p} \frac{\Pd_{XP}(0^n,p)}{p}
 \approx \frac{1/4}{1/2} + \frac{1/2}{1} = 1
$$
And indeed, this property holds for any $x\in \{0,1\}^n$ and any strategy where the prover first chooses some distribution $\Pd_i$ with probability $q_i$, and then behaves honestly for $\Pd_i$. Let $I$ be the random variable over the choice of $i$ indexing the distributions $\Pd_i$. We find
\begin{align*}
	\sum_{p} \frac{\Pd_{XP}(x,p)}{p}
	&= \sum_{p} \sum_i 
			\frac{q_i \Pr[(X,P)=(x,p) | I=i]}{p} \\
	&= \sum_i q_i \sum_{p} 
			\frac{\Pr[(X,P)=(x,p) | I=i]}{p} \\
	&=	\sum_{i: \Pd_i(x)>0} q_i \frac{\Pr[(X,P)=(x,\Pd_i(x)) | I=i]}
		{\Pd_i(x)}
	\approx 
		\sum_{i: \Pd_i(x)>0} q_i \leq 1.
\end{align*}
If the prover decides to make the verifier reject with some probability, then the above sum decreases. 

\paragraph{Remark 4: Soundness does not necessarily imply a convex combination of provers.}
Consider the following distribution $\Pd_{XP}$: $\Pd_{XP}(x_1, 1/2)=1/4$, $\Pd_{XP}(x_1, 1/4)=1/8$, $\Pd_{XP}(x_2, 3/4)=1/2$, and $\Pd_{XP}(x_2, 3/8)=1/8$. Clearly $\Pd_{XP}$ is a probability distribution, and the soundness condition is satisfied, i.e.~for $x\in \{x_1, x_2\}$ we have $\sum_p \frac{\Pd_{XP}(x,p)}{p}=1$. 
Then the prover who first chooses a probability distribution and then behaves honestly for it must choose with positive probability a distribution $\Pd_i$ with $\Pd_i(x_1)=1/2$. (Otherwise, $(x_1, 1/2)$ would have probability zero). Then $\Pd_i(x_2)=1/2$, and thus outputting $(x_2, 1/2)$ must also have positive probability, but $\Pd_{XP}(x_2,1/2)=0$. 

\subsection{Technical Overview}
\label{sec:TechnicalOverviewSampling}
We informally describe a simplified version of the sampling protocol of Theorem~\ref{thm:SamplingInformal}, and sketch how correctness and soundness can be proved. Then we discuss how to get rid of the simplifying assumptions. 

\paragraph{Histograms.} 
Let $\Pd$ be a distribution over $\{0,1\}^n$, and consider
the $(1,n)$-histogram of $\Pd$ according to Definition~\ref{def:histogram}. That is, we let $h = (h_0, \ldots, h_n)$, where $h_i := \Pr_{y\leftarrow \Pd}[y \in \cB_i]$ for $\cB_i := \left\{x: \Pd(x) \in (2^{-(i+1)}, 2^{-i}]\right\}$. For simplicity we assume that for all $x$ either $\Pd(x)=0$ or $\Pd(x)\geq 2^{-n}$, which implies that $\sum_{i\in (n)} h_i = 1$. 

\paragraph{A simplified sampling protocol with an inefficient verifier.} 
If efficiency were not an issue, the honest prover could just send all pairs $(x,\Pd(x))$ to the verifier, who then outputs $(x,\Pd(x))$ with probability $\Pd(x)$. It is clear that this protocol achieves even stronger completeness and soundness guarantees than stated in our theorem. 
We now change this protocol, still leaving the verifier inefficient. But, using hashing, the verifier of this modified protocol can later be made efficient.

To describe the intuition, we make the following simplifications: we let the verifier output probabilities of the form $2^{-j}$, and when interacting with the honest prover, the verifier will output pairs $(x,p)$ where $p$ is a $2$-approximation of $\Pd(x)$. We will also make an assumption on $\Pd$, but it is easiest to state it while describing the protocol. 

\vspace{0.2cm}
\noindent\textit{The protocol.}
The honest prover sends the histogram $h$ of $\Pd$ to the verifier. The verifier splits the interval $[0, n]$ into intervals $\cJ_j$ of length $\log_2(n)$. We denote by $\cI_i:=\cJ_{2i}$ the even intervals, and we will call the odd intervals \textit{gaps}. For simplicity we assume $\log_2(n)$ is an even integer, and that $\sum_k \sum_{i\in \cI_k} h_i = 1$, i.e. that $\Pd$ is such that $h$ has no probability mass in the gaps. Now the verifier selects an interval $\cI_k$ at random, where the probability of $\cI_k$ corresponds to its weight according to $h$, i.e.~the probability of $\cI_k$ is $w_k := \sum_{j\in \cI_k}h_j$. The prover sends sets $\cX_i$ for $i\in \cI_k$ to the verifier, where the honest prover lets $\cX_i = \cB_i$. The verifier checks that the $\cX_i$ are disjoint, and that $|\cX_i| \in 2^{\pm 1} h_i2^i$.\footnote{
	Actually, the verifier can check the stronger condition
	$|\cX_i| \in [1/2,1] h_i2^i$, but to simplify our
	statements we use	$2^{\pm 1}$. 
} It then randomly chooses one of the sets $\cX_j$, where $\cX_j$ has probability $\frac{h_j}{\sum_{i\in \cI_k} h_j}$. Finally, the verifier chooses a uniform random element $x$ from $\cX_j$, and outputs $(x,2^{-j})$. 

\vspace{0.2cm}
\noindent\textit{Completeness:} 
It is not hard to see that if the prover is honest, then for any $x$ and $j$ the following holds. If $x \in \cB_j$, then 
$\Pd_{XP}(x,2^{-j}) \in 2^{\pm 1}\Pd(x)$. Otherwise, $\Pd_{XP}(x,2^{-j})=0$. 

\vspace{0.2cm}
\noindent\textit{Soundness:}
We sketch a proof for the following soundness guarantee: for any (possibly dishonest) prover there is an event $\Bad$ such that 
$\Pr[\Bad]\leq 50/\sqrt{n}$, and for all $x$ we have
$\sum_p \frac{\Pr[(X,P)=(x,p)|\lnot \Bad]}{p} \leq 2+50/\sqrt{n}$. 

We let $k(j)$ be the function that outputs the interval of $j$ (i.e.~$k$ such that $j\in \cI_k$). We first observe that\footnote{
	In Eqn.~(\ref{eqn:SampOvw1}) we implicitly assume that the prover always sends disjoint sets $\cX_j$ that are of appropriate size. 
	Removing this assumption is a minor technicality that is dealt with in the full proof. 
}
\begin{align}
	\Pd_{XP}(x,2^{-j}) \in 2^{\pm 1} \cdot \Pr[x \in \cX_j | k=k(j)] \cdot 2^{-j},
	\label{eqn:SampOvw1}
\end{align}
where $\Pr[x \in \cX_j | k=k(j)]$ is the probability that the prover puts $x$ into $\cX_j$ given that $k(j)$ was chosen by the verifier. 
To see this, note that 
$
	\Pd_{XP}(x,2^{-j}) = \Pr[x \in \cX_j | k=k(j)]
								\cdot w_k
								\cdot \frac{h_j}{\sum_{i \in \cI_k}h_i}
								\cdot \frac{1}{|\cX_j|}
$,
where $w_k$ is the probability of choosing $k$, the third factor is the probability of choosing $j$ given $k$ was chosen, and $1/|\cX_j|$ is the probability of choosing $x$, given $k,j$ were chosen and $x$ is in $\cX_j$. Indeed, by the definition of $w_k$ and using the verifier's check $|\cX_j| \in 2^{\pm 1} h_j2^j$, this implies (\ref{eqn:SampOvw1}).

To prove our soundness claims, we consider the following sets representing small, medium and large probabilities, respectively.
\begin{align*}
	\cS(x) &= \left\{j: 2^{-j} \leq \frac{\Pd_X(x)}{\sqrt{n}}\right\}
				 = \left\{j: j\geq \log_2(1/\Pd_{X}(x))+1/2\log_2(n)\right\},\\
	\cM(x) &= \left\{j: 2^{-j} > \frac{\Pd_{X}(x)}{\sqrt{n}}
								\land 
								2^{-j} < \sqrt{n}\Pd_{X}(x)\right\}\\
				 &= \left\{j: \log_2(1/\Pd_{X}(x))-1/2\log_2(n) 
								< j < 
								\log_2(1/\Pd_{X}(x))+1/2\log_2(n)\right\},\\
	\cL(x) &= \left\{j: 2^{-j} \geq \sqrt{n}\Pd_{X}(x)\}
				 = \{j: j \leq \log_2(1/\Pd_{X}(x))-1/2\log_2(n)\right\}.
\end{align*}
We define the event $\Bad$ to occur if the verifier outputs a probability that is much too small, i.e.~it outputs some $(x,2^{-j})$ where $j\in \cS(x)$. We find
\begin{align*}
	\Pr[\Bad] &= \sum_x\sum_{j\in \cS(x)} \Pd_{XP}(x,2^{-j})
		\stackrel{\text{(\ref{eqn:SampOvw1})}}{\leq }
			\sum_x\sum_{j\in \cS(x)}
				2 \cdot \Pr[x \in \cX_j | k=k(j)] \cdot 2^{-j} \\
		&\leq 
			2 \sum_x\sum_{j\in \cS(x)}
				 2^{-j}
		\leq 
			2\sum_x \frac{\Pd_{X}(x)}{\sqrt{n}} 
			\sum_{i=0}^{\infty}\frac{1}{2^i}
		\leq 
			\frac{4}{\sqrt{n}},
\end{align*}
where the third inequality follows by definition of $\cS$. 

To prove the second soundness claim, we find
$
	\sum_p \frac{\Pr[(X,P)=(x,p)\land \lnot \Bad]}{p}
	=
		\sum_{j \in \cM(x)} 
		\frac{\Pd_{XP}(x,2^{-j})}{2^{-j}}
	+
		\sum_{j \in \cL(x)}
		\frac{\Pd_{XP}(x,2^{-j})}{2^{-j}} 
$.
Now 
\begin{align*}
		&\sum_{j \in \cL(x)}
			\frac{\Pd_{XP}(x,2^{-j})}{2^{-j}}
			\leq
			\frac{1}{\sqrt{n}\Pd_{X}(x)} 
				\underbrace{\sum_{j \in \cL(x)}
				\Pd_{XP}(x,2^{-j})}_{\leq \Pd_{X}(x)}
			\leq \frac{1}{\sqrt{n}},		\\
		&\sum_{j \in \cM(x)}
			\frac{\Pd_{XP}(x,2^{-j})}{2^{-j}}
		\stackrel{\text{(\ref{eqn:SampOvw1})}}{\leq }
			2 \sum_{j \in \cM(x)} \Pr[x \in \cX_j | k=k(j)]
		\leq 2,
\end{align*}
where the last inequality again follows since the prover must send disjoint sets $\cX_i$, and by definition of $\cM(x)$ we have $|\cM(x)| < \log_2(n)$ and thus $\cM(x) \cap \cI_k$ is non-empty for at most one $\cI_k$. (This is one reason for defining the intervals, gaps and the sets $\cS, \cM, \cL$ as we did!) Thus we have 
$
	\sum_p \frac{\Pr[(X,P)=(x,p)| \lnot \Bad]}{p}
	=
	\frac{1}{\Pr[\lnot \Bad]}
		\sum_p \frac{\Pr[(X,P)=(x,p)\land \lnot \Bad]}{p}
	\leq (1+8/\sqrt{n})(2+1/\sqrt{n}) \leq 2+17/\sqrt{n}$.

\paragraph{Making the verifier efficient.}
In the protocol as sketched above, it is possible to make the verifier run in time $\poly(n)$. To achieve this, the prover does not send the whole sets $\cB_j$, but instead the verifier chooses a $3$-wise independent hash function 
$f$ mapping $n$ bits to $m$ bits, where 
$m = \log_2(\sum_{i\in \cI_k} 2^i h_i) + c \log_2(n)$.
The hash mixing lemma (as given in Section~\ref{sec:HashFunctions}) ensures that there are close to $n^c$ many elements $x\in \bigcup_{i\in \cI_k}\cB_i$ with $f(x)=0^m$.\footnote{
	We use $3$-wise (and not $2$-wise) independent hashing because in the analysis we need to argue that $\{x': x' \in \bigcup_{i\in \cI_k}\cB_i \land f(x')=0^m\}$ is of size close to $n^c$ even conditioned on the event $f(x)=0$ for a fixed $x$. 
} 
The prover now only sends these elements, and the verifier checks that the set sizes are appropriate. It is possible to still achieve the same soundness and completeness guarantees, and we remark that we specifically designed the protocol so that this is possible.

\paragraph{Using a more accurate histogram.}
Using the basis $2^{\eps}$ instead of $2$ for the histogram, where $\eps$ is polynomially small allows to improve the protocol as follows. For completeness, there exists a constant $c$ such that each $x$ is output with probability $2^{\pm c\eps} \Pd(x)$. For soundness, there exist constants $c$ and $d$ such that $\Pr[\Bad] \leq c\eps$ and
$\sum_p \frac{\Pr[(X,P)=(x,p)|\lnot \Bad]}{p} \leq 1+d\eps$. 

\paragraph{Handling general distributions.}
Above we made the simplifying assumption that there is no probability mass in the gaps. That is, we assumed $\sum_k\sum_{i\in \cI_k}h_i = 1$. This assumption clearly does not hold in general. If we analyse the given protocol assuming that 
$\sum_k\sum_{i\in \cI_k}h_i = w$ for some $w \in [0,1]$, the completeness and soundness guarantees we get change as follows:

\vspace{0.2cm}
\noindent\textit{Completeness:} There exists a constant $c$ such that for any $x$ and $j$ the following holds. If $x\in \cB_j$ and $j\in \bigcup_i \cI_i$, then $\Pd_{XP}(x,2^{-j}) \in \frac{2^{\pm c\eps}\Pd(x)}{w}$. Otherwise, 
$\Pd_{XP}(x,2^{-j})= 0$. 

\vspace{0.2cm}
\noindent\textit{Soundness:}
There exist constants $c$ and $d$ such that the following holds. 
For any (possibly dishonest) prover there is an event $\Bad$ such that $\Pr[\Bad]\leq c\eps$, and for all $x$ we have $\sum_p \frac{\Pr[(X,P)=(x,p)|\lnot \Bad]}{p} \leq \frac{1+d\eps}{w}$. 

\vspace{0.2cm}
\noindent Note that by averaging, the verifier can place the gaps such that $w > \frac{1}{2}$.

For the final protocol, we want that any $x$ has a chance of being output. For this, the verifier considers both possible gap placements, i.e.~it either switches the intervals and gaps, or not by choosing a random bit $s\in \{0,1\}$, where each gap placement is chosen with probability $w(s)$, which is the weight inside the intervals. For each fixed choice of 
$s$ the above analysis applies. This gives that each element is output with the correct probability in case the prover is honest. However, one only obtains the soundness guarantee $\sum_p \frac{\Pr[(X,P)=(x,p)|\lnot \Bad]}{p} \leq 2(1+d\eps)$.

This factor of $2$ can be decreased to $(1+\delta)$ by making the gaps smaller, such that the gap size equals $\delta$ times the interval size. Unfortunately there is a tradeoff: 
If the gap size is $o(\log(n))$, then the event $\Bad$ no longer has polynomially small probability. Furthermore, if we fix the gap size to $\Theta(\log(n))$, and set $\delta = o(1)$, then the verifier is no longer efficient because the prover needs to send too many elements. For the exact tradeoff, we refer to Theorem~\ref{thm:main}.

\paragraph{Outputting the exact probabilities.}
To improve the completeness such that the exact probability $\Pd(x)$ (instead of $2^{-j}$) is output, the verifier simply asks the prover to send the correct probability $\Pd(x)\in (2^{-(j+1)}, 2^{-j}]$ in the end. 

\subsection{Theorem Statement}
\label{sec:SamplingTheoremStatement}
We now state our main theorem. For an informal discussion of the theorem statement and a proof sketch we refer to Section~\ref{sec:TechnicalOverviewSampling}. In our sampling protocol, the verifier will output pairs $(x,p) \in \{0,1\}^n \times (0,1]$. For a fixed prover, we let $(X,P)$ be the random variables corresponding to the verifier's output, and we denote their joint distribution by $\Pd_{XP}$. Also, $\Pd_X$ is defined by $\Pd_X(x)=\sum_p \Pd_{XP}(x,p)$. 
\begin{theorem}[The Sampling Protocol]
	\label{thm:main}
	There exists a constant-round public-coin interactive
	protocol such that the following holds.
	The verifier and the prover take as input 
	$n \in \mathbb{N}, \eps, \delta \in (0,1)$, 
	and the prover additionally gets as input a probability
	distribution $\Pd$ over $\{0,1\}^n$.
	
	The verifier runs in time
	$\poly\left(n \left(\frac{1}{\eps}\right)^{1/\delta}\right)$
	and either outputs $(x,p) \in \{0,1\}^n	\times (0,1]$, or
	rejects. Furthermore, we have:
	
	\vspace{0.2cm}	
	\noindent\textbf{Completeness: }
	Suppose the prover is honest. 
	Then the verifier rejects with probability at most $\eps$.
	Furthermore, there exists a set $\cM \subseteq \{0,1\}^n$ with
	$\Pd(\cM) \geq 1-\eps$ such that	for all $x\in \{0,1\}^n$ we have
	\begin{align*}	
		\Pd_{XP}(x,p)
			\begin{cases}
				\in (1\pm \eps)\cdot \Pd(x)
											& \text{if } p = \Pd(x) \land x \in \cM \\
				= 0					& \text{otherwise.} 
		\end{cases}
	\end{align*}

	\noindent\textbf{Soundness: }
	Fix any deterministic prover. Then there is an event $\Bad$ such that 
	the following holds:
	\begin{align*}
		\text{(i)}  \quad &\Pr[\Bad]	\leq \eps, \\
		\text{(ii)} \quad &\text{For all $x \in \{0,1\}^n$ we have} 
					\sum_{p\in (0,1]} \frac{\Pr[(X,P)=(x,p)|\lnot \Bad]}{p} \leq 1+\eps +\delta.
	\end{align*}
\end{theorem}

We give a protocol for the case 
$\bigl(\frac{9000}{\eps}\bigr)^{16/\delta}\leq 2^{n/50}$. In case this inequality
does not hold, we only need to give a protocol that runs in time $\poly(2^n)$. The protocol where 
the honest prover sends 
$\{(x, \Pd(x)): x\in \{0,1\}^n\}$,
and the verifier outputs $(x,p)$ with probability $p$ gives 
stronger guarantees than what is required by the theorem, 
and clearly runs within the required time bound.

\subsection{The Protocol}
\label{sec:SamplingProtocolDescription}

\paragraph{Definitions and protocol parameters.}
Let $\cH(n,m)$ be a family of $3$-wise independent hash functions mapping $n$ bits to $m$ bits.
Given a finite set $\cJ \subseteq \mathbb{N}$, and a set 
$\{ (w_i,i) \}_{i\in \cJ}$ where $w_i \in [0,1]$, let
$\cW(\{ (w_i,i) \}_{i\in \cJ})$ be the distribution that outputs $i$ with probability $\frac{w_i}{\sum_{j\in \cJ}w_j}$.

Given an input $(n,\eps',\delta',\Pd)$, we define $\eps:=\eps'/9000$ and $\delta:=\delta'/16$. Our assumption becomes $\bigl(\frac{1}{\eps}\bigr)^{1/\delta}\leq 2^{n/50}$. We let $t:=\left \lceil \frac{2n}{\eps} \right\rceil$. 
For $i \in (t)$, let $\cA_i$ and $\cB_i$ be defined as in 
Definition~\ref{def:histogram} for $\Pd$. We let
\begin{align*}
	&\Gsize':= \frac{2}{\eps}\log_2\left(\frac{1}{\eps}\right),
	\quad 
	\Isize':= \frac{\Gsize'}{\delta},\\
	&\Gsize:=\left\lceil\Gsize'\right\rceil, 
	\quad
	\Isize:=\left\lceil\frac{\Isize'}{\Gsize}
		\right\rceil\Gsize, \\
	&\SampGap:=\log_2\Bigl(
				\frac{t \cdot
							\Isize' \cdot 2^{2\cdot\Isize'\eps}}
						 {\eps^4}
			\Bigr).
\end{align*}
We let $\cK:=\Bigl\{0, 1, \ldots,
				\bigl\lceil\frac{t}
				{\Isize+\Gsize}\bigr\rceil \Bigr\}$.
		For $i \in \cK\setminus \{0\}$ and shift 
		$s \in \cS := \{-1, \Gsize -1, 2\cdot \Gsize -1, \ldots, 
				(\frac{\Isize}{\Gsize}-1)\cdot \Gsize-1\}$ we 
define the intervals
		\begin{align*}
			\cI_0(s) 			:= \{&0,\ldots, s\} \\
			\cI_{i}(s)	:= \bigl\{&s + i \Gsize+ (i-1) 
											\Isize+1, \ldots, \\
								     &s + i(\Gsize +
								     \Isize)\bigr\} \cap (t).
		\end{align*}
The regions between these intervals will be called \textit{gaps}. 
\begin{claim}
	\label{claim:ParamProperties}
	The following properties hold:
	\begin{enumerate}[(i)]
	\item $\Isize,\Gsize,\frac{\Isize}{\Gsize} \in \mathbb{N}$ and $\Gsize, \Isize \geq 1$. 
	\item For all $i$ and $s$ we have $|\cI_i(s)| \leq 
				\Isize$. The gaps have size exactly $\Gsize$. 				
	\item $\Isize' \leq \Isize \leq 2\cdot \Isize'$, and thus
				$2^{\SampGap}$ is contained in the interval
			\begin{align*}
				\left[\frac{t \cdot \Isize \cdot 2^{\Isize\cdot \eps}}
						 {2\eps^4},
					\frac{t \cdot \Isize \cdot 
							2^{2\cdot \Isize \cdot \eps}}{\eps^4}\right]
			\end{align*}
		\item There are at least $25$ intervals and gaps: 
					$|\cK| \geq 25$.
	\end{enumerate}
\end{claim}
\begin{proof}
	Parts (i)-(iii) follow directly by definition. 
	To see (iv), note that 
	$\bigl(\frac{1}{\eps}\bigr)^{1/\delta}\leq 2^{n/50}$, 
	implies
	$
		\frac{\Isize}{2} 
		\leq \Isize' 
		= \frac{1}{\eps\delta}\log_2\left(\frac{1}{\eps}\right)
		\leq \frac{t}{100}
	$. Thus, since $\Isize \geq \Gsize$, we have that 
	$|\cK| \geq \frac{t}{\Isize+\Gsize} \geq 
	\frac{t}{2\Isize} \geq \frac{t}{4\Isize'}\geq 25$.
\end{proof}	

\paragraph{The Sampling Protocol.} On input $(n,\eps',\delta')$ to the verifier, and $(n,\eps',\delta',\Pd)$ to the prover:

\begin{itemize}			
	\item[] \textbf{Prover:} 
		Send an $(\eps, t)$-histogram $h$ to the verifier.
		
		If the prover is honest, it lets 
		$h$ be the $(\eps, t)$-histogram of $\Pd$.
	\item[] \textbf{Verifier:} 
		Reject if $\sum_{i\in (t)} h_{i} \notin [1-2^{-n},1]$. 
		Define 
		$\cN := \{j: h_j \geq \frac{\eps}{2t}\}$.
		
		For $i \in \cK$ let $w_{i}(s) := \sum_{j\in \cI_{i}(s)} h_j$, 
		and define $w(s):=\sum_{i \in \cK}w_i(s)$.
		
		Pick $s \inu \cW\bigl(\{(w(s),s)\}_{s\in \cS}\bigr)$.
		
		For all $i \in \cK$ define $\cI_{i} := \cI_{i}(s)$,
		$w_{i}:=w_{i}(s)$, $w := w(s)$ and 
		$\cI'_{i}:= \cI_i \cap \cN$.
		
		Pick $k \inu \cW\bigl( \{(w_i,i)\}_{i\in \cK} \bigr)$, and let
		\begin{align*}
			g:= \SampGap + 
							\Bigl(&\bigl(\log_2(\sum_{i\in \cI_k}2^{i\eps}
							h_i)-\SampGap\bigr) \\
							&- 
							\bigl\lfloor\log_2(\sum_{i\in \cI_k}2^{i\eps}
							h_i)-\SampGap \bigr\rfloor\Bigr)
		\end{align*}		
		Let $m:=\max\bigl\{0,\log_2\bigl(\sum_{i\in \cI_k}2^{i\eps}
						 h_i\bigr)-g\bigr\}$ 
		(By definition of $g$, $m$ is integer)
		
		Pick $f \inu \cH(n,m)$, and send $(s,k,f)$ to the prover.
	\item[] \textbf{Prover:} Send sets 
		$\{\cX_i\}_{i\in \cI'_k}$ to the verifier.
	
		If the prover is honest, it lets $ \cX_i := \{x \in \cB_i: f(x) = 0^m\}$.
	\item[] \textbf{Verifier:}
		If one of the following conditions does not hold, reject:
		\begin{align}
		  \text{(a)} \quad 
		  	&\forall i\in \cI'_k: \forall x \in \cX_i: f(x) = 0^m \nonumber \\
			\text{(b)} \quad
				&\text{If $m = 0$ then} \quad  
				\forall i\in \cI'_k: 
					2^{i\eps} h_i
					\leq |\cX_{i}| \leq
					2^{(i+1)\eps} h_i
				\nonumber	\\
				&\text{If $m > 0$ then} \quad
				\forall i\in \cI'_k: 
					2^{-\eps} 
						\frac{2^{g} 2^{i\eps} h_i}
								 {\sum_{j\in \cI_{k}}2^{j\eps}h_j}
					\leq |\cX_{i}| \leq
					2^{\eps} 
						\frac{2^{g} 2^{(i+1)\eps} h_i}
								 {\sum_{j\in \cI_{k}}2^{j\eps}h_j}
				\nonumber \\
			\text{(c)} \quad &\text{The sets $\cX_i$ are pairwise disjoint}\nonumber
		\end{align}
		Pick $j \inu \cW\bigl(\{(h_i, i)\}_{i\in \cI_k}\bigr)$.
		Reject if $j \notin \cI'_k$.
		Pick $x \inu \cX_j$, and 
		send	$(j,x)$ to the prover.			
\item[] \textbf{Prover:}
		Send $p$ to the verifier. 
		
		If the prover is honest, it sends $p:=\Pd(x)$.
		
\item[] \textbf{Verifier:}
		If $p \notin \cA_j$,
		let $p:= 2^{-j\eps}$.
		Output $(x,p)$.
\end{itemize}

\subsection{Analysis of the Protocol}
\label{sec:SamplingProtocolAnalysis}
Throughout this section we assume $\bigl(\frac{1}{\eps}\bigr)^{1/\delta}\leq 2^{n/50}$. 
For a fixed deterministic prover, we let $S, K, F, J, X, P$ be the random variables corresponding 
to $s,k,f,j,x,p$ in the protocol. 
For $s\in \cS$ we use the notation
$\Pd_{XP\mid S=s}(x,p)= \Pr[(X,P)=(x,p)|S=s]$ and  
$\Pd_{X\mid S=s}(x) = \Pr[X = x| S=s]$.
We will prove the following statements about the sampling protocol, which describe the guarantees we get conditioned on a fixed choice of the gaps, as given by shift $s$. 

\begin{lemma}[Running time]
	\label{lem:runningTime}
	For any deterministic prover the verifier runs in time 
	$\poly\left(n \left(\frac{1}{\eps}\right)^{1/\delta
	}\right)$.
\end{lemma}

\begin{lemma}[Completeness for fixed $s$]
	\label{lem:complFixedS}
	Suppose the prover is honest.
	Then the verifier rejects with probability at most $25\eps$.
	Furthermore, there exists a set $\cM \subseteq \{0,1\}^n$
	with $\Pd(\cM) \geq 1-2\eps$ such that for every $s\in \cS$	
	there exists a set $\cM_s \subseteq \{0,1\}^n$ with 
	$w(s) := \Pd(\cM_s)$ such that for all $x\in \{0,1\}^n$ we
	have
	\begin{align*}
		\Pd_{XP\mid S=s}(x,p) 
			\begin{cases}
				\in(1\pm 132\eps)\cdot \frac{\Pd(x)}{w(s)}
											& \text{if } p = \Pd(x) \land x\in \cM \cap \cM_s,\\
		 = 0   					& \text{otherwise.} 
		\end{cases}
	\end{align*}
	Furthermore, the sets $\{\overline{\cM}_s \cap \cM\}_{s\in \cS}$ form a 
	partition of $\cM$. 
\end{lemma}

\begin{lemma}[Soundness for fixed $s$]
	\label{lem:soundFixedS}
	Fix any deterministic prover and $s \in \cS$.
	There exists a function
	$f_s: \{0,1\}^n \rightarrow [0,1]$ such that
	\begin{align*}
		\text{(i)}  \qquad &\Pr[P\leq f_s(X)|S=s] 
								\leq \frac{16\eps}{w(s)}, \\
		\text{(ii)} \qquad &\text{For all $x \in \{0,1\}^n$ we have} 
					\sum_{p> f_s(x)} \frac{\Pd_{XP\mid S=s}(x,p)}{p} \leq \frac{(1+6\eps)}{w(s)}.
	\end{align*}
\end{lemma}
 
Given these Lemmas, it is not difficult to prove 
Theorem~\ref{thm:main}. We defer the detailed proof to 
Section~\ref{sec:ProofOfMainSamplingTheorem}, and 
only give a proof sketch.

\begin{proof}[Proof of Theorem~\ref{thm:main} (Sketch)]
	The running time is given by 
	Lemma~\ref{lem:runningTime}. Defining the event $\Bad$ as 
	$P\leq f_S(X)$, completeness and soundness follow by 
	conditioning on $s\in \cS$ and applying 
	Lemmas~\ref{lem:complFixedS} and \ref{lem:soundFixedS},
	respectively.	
\end{proof}

It remains to prove the three lemmas above. The verifier's 
running time is easy to bound, as we just need to bound $t$ 
and the size of the sets $\cX_i$. The proof can be found 
in Section~\ref{sec:VerifiersRunningTime}.

To prove completeness and soundness, it is useful to consider
the probability that the prover puts $x$ into the set $\cX_i$
given that $f(x)=0$. Formally, for a fixed $s\in \cS$ we define 
the function
	$k_s: \bigcup_{i \in \cK} \cI_{i}(s) \rightarrow \cK$ that maps 
	$j \mapsto i$ such that $j \in \cI_{i}(s)$. 
Given a deterministic prover, for any $(x,j) \in \{0,1\}^n
	\times (t)$ define $r_s(x,j) :=$
	\begin{align*}
		\begin{cases} 
		 		\Pr_F\Bigl[(a)-(c) \land
					x \in \cX_{j} 
							\Big|(K,F(x),S)=(k_s(j),0^m,s)\Bigr]
				    &\mbox{if } 
				    	j \in \bigcup_{\ell \in \cK} \cI'_{\ell} \\
				  0
				    &\mbox{otherwise.}
		\end{cases}
	\end{align*}
	We note that once a deterministic prover is fixed, 
	the above probability is only over the choice of $f$.
	Also, we are only interested in the cases where (a)-(c) 
	hold, as otherwise the verifier rejects. 
We can then relate $\Pd_{XP\mid S=s}(x,p)$ and $r_s(x,p)$ as follows:
\begin{lemma}
	\label{lemma:qvsr}
	Fix any deterministic prover and $s\in \cS$. 
	Then for any $(x,j) \in \{0,1\}^n \times (t)$ we have
	\begin{align*}
		\sum_{p\in \cA_j}\Pd_{XP\mid S=s}(x,p) \in 
			\left[
				2^{-2\eps}\frac{r_s(x,j)}
								{w(s)\cdot 2^{j\eps}},
				2^{\eps}
					\frac{r_s(x,j)}
							 {w(s)\cdot 2^{j\eps}}
			\right].
		\end{align*}
\end{lemma}
A formal proof can be found in 
Section~\ref{sec:BoundingProbOfpInAj}. We provide a proof sketch.
\begin{proof}[Proof (Sketch)]
	The sum in the lemma equals $\Pr_{X,J}[(X,J)=(x,j)|S=s]$, which 
	equals
	\begin{align*}
		\frac{w_k}{w(s)}
			\cdot
		2^{-m}
		  \cdot
		\frac{h_j}{\sum_{i\in \cI_k} h_i}
		  \cdot
		r_s(x,j)
			\cdot
		\frac{1}{|\cX_j|},
	\end{align*}
	where the first factor is the probability that the appropriate
	$k$ is chosen in the protocol, the second factor is the probability
	that $f(x)=0$, the third factor is the probability that $j$ is 
	chosen (conditioned on $K=k$), 
	the fourth factor is the probability that the prover
	puts $x$ into $\cX_j$ (conditioned on $(K,F(x),J)=(k,0^m,j)$),
	and the fifth factor is the probability that 
	$x$ is chosen from $\cX_j$ 
	(conditioned on $(K,F(x),J)=(k,0^m,j) \land x\in \cX_j$). 
	The lemma then follows by definition of $m$ and 
	the verifier's check (b) that specifies	the size of the sets 
	$\cX_i$.
\end{proof}

\subsubsection{Proof of completeness: overview}
The goal of this section is to prove Lemma~\ref{lem:complFixedS}.
We only give proof sketches here and defer the formal proofs to 
Section~\ref{sec:SamplingProofCompletenessDetails}. 
Throughout this section we assume the prover is honest.
We prove completeness for $\cM := \bigcup_{i \in \cN} \cB_i$ 
and 
$$\cM_s:= \bigcup_{i \in \bigcup_{j\in \cK} \cI_j(s)}\cB_i \qquad \text{
		for $s\in \cS$}.$$ 
As $t$ is chosen such that only probabilities smaller than $2^{-2n}$
are neglected, we get $\sum_{i\in (t)}h_i \geq 1-2^{-n}$ (see
Claim~\ref{claim:sumwilarge}). Thus, the verifier never rejects after
receiving the first message. 
Also, it is easy to see that $\Pd(\cM) \geq 1-2\eps$ and 
that the sets $\{\overline{\cM}_s \cap \cM\}_{s\in \cS}$ 
form a partition of $\cM$ (see Claim~\ref{claim:propertiesOfM}). 

Now since the protocol only considers $h_i \geq \frac{\eps}{2t}$ (i.e.~$i\in \cN$), and by Claim~\ref{claim:sizeBi} we have 
$|\cB_i|\geq h_i2^{i\eps}$, the following lower bound on $|\cB_i|$ for
$i\in \cN$ follows (a formal proof can be found in 
Section~\ref{sec:SamplingProofCompletenessDetails}):
\begin{claim}
	\label{claim:BareBig}
	Fix any $s\in \cS$ and consider any protocol execution where 
	$(K,S)=(k,s)$	and $m>0$. Then for all $i\in \cI_k'$ we have
	$
			\frac{|\cB_i|}{2^m} \geq 
			 \frac{\Isize}
						{4\eps^3}
	$.
\end{claim}
Next, we show that the verifier rejects with small probability only. 
\begin{lemma}
	\label{lem:RejectWithSmallProb}
	The verifier rejects with probability at most $25\eps$.
\end{lemma}
We provide a proof sketch, and the full proof can again be found in Section~\ref{sec:SamplingProofCompletenessDetails}.
\begin{proof}[Proof (Sketch)]
	First, it is easy to see that the verifier only rejects with 
	small probability due to $j\notin \cN$, as by choice of $t$ such
	$j$ is only chosen with small probability.
	
	It remains to show that the verifier rejects in (a)-(c) only with
	small probability. Clearly, (c) always holds. The
	case $m=0$ is trivial, so suppose $m>0$. 
	Then by definition of $m$ and Claim~\ref{claim:sizeBi}, we find
	\begin{align}
		\frac{2^g2^{i\eps}h_i}{\sum_{\ell \in \cI_k}2^{\ell\eps}
		 h_{\ell}} \leq \frac{|\cB_i|}{2^m} \leq 
		 \frac{2^g2^{(i+1)\eps}h_i}{\sum_{\ell \in \cI_k}2^{\ell\eps}
		 h_{\ell}}, \label{eqn:24}
	\end{align}
	and thus we can apply the hash mixing lemma (Lemma~\ref{lemma:hashMixingLemma}) part (i), and 
	then Claim~\ref{claim:BareBig} to find 
	that (b) holds with high probability.
\end{proof}

Finally, we can prove completeness for a fixed value of $s$.
Again, the formal proof can be found in Section~\ref{sec:SamplingProofCompletenessDetails}. 
\begin{proof}[Proof of Lemma~\ref{lem:complFixedS} (Sketch)]	
	Since the prover is honest, we have $r(x,j)=0$ in case
	$x\notin \cB_j$ or $j\notin \cN$. 
	We first show that for any $x \in \cB_j$ with $j\in \cN$ we have 
	that $r_s(x,j)$ is close to $1$.
	In case $m=0$ we clearly have $r_s(x,j)=1$, so assume $m>0$.
	Since the prover is honest, we get that
	$r_s(x,j)=\Pr_F[(b)|(K,F(x),S)=(k_s(j),0^m,s)]$. 
	Assuming $f(x)=0$, $x$ is contained in exactly one set 
	$\cX_{i^*}$. Now again using (\ref{eqn:24}), 
	the claim	follows by first applying the hash mixing lemma 
	(Lemma~\ref{lemma:hashMixingLemma})
	(part (ii) in case $i\neq i^*$, and part (iii) in case 
	$i=i^*$), and then Claim~\ref{claim:BareBig}. 
	
	Finally, since the honest prover always sends $p=\Pd(x)$, 
	Lemma~\ref{lemma:qvsr} gives that $\Pd_{XP}(x,\Pd(x))$ is close to 
	$\frac{2^{-j\eps}}{w(s)}$, which in turn is close to 
	$\frac{\Pd(x)}{w(s)}$ as $x \in \cB_j$. 
\end{proof}

\subsubsection{Proof of soundness: overview}
The goal of this section is to prove Lemma~\ref{lem:soundFixedS}.
Throughout this subsection, fix any deterministic prover.
We first give the following simple claim: 
since the verifier ensures
that the sets $\cX_j$ sent by the prover are disjoint, for any
$s,k$ the sum over $j \in \cI_k(s)$ of the conditional probabilities $r_s(x,j)$ is at most $1$. A formal proof can be found in Section~\ref{sec:SamplingProofSoundnessDetails}.
\begin{claim}
	\label{claim:rsumtoone}
	For any $(s,k, x)\in \cS \times \cK \times \{0,1\}^n$ we have 
	$\sum_{j \in \cI_k(s)}r_s(x,j) \leq 1$.
\end{claim} 

Fix $s\in \cS$. We prove soundness for the function
	$f_s(x):=\frac{\Pd_{X\mid S=s}(x)}{2^{(\Gsize/2-1)\eps}}$.
	For a fixed $x\in \{0,1\}^n$ we define sets $\cS(x),\cM(x),\cL(x) 
	\subseteq [0,1]$ representing small, medium and large
	probabilities as follows:
	\begin{align*}
		\cS(x)&:= \left[0,\frac{\Pd_{X\mid S=s}(x)}{2^{(\Gsize/2-1)\eps}}\right],\\
		\cM(x)&:=\left(\frac{\Pd_{X\mid S=s}(x)}{2^{(\Gsize/2-1)\eps}},
							2^{(\Gsize/2-1)\eps}\Pd_{X\mid S=s}(x)\right), \\
		\cL(x)&:=\left[2^{(\Gsize/2-1)\eps}\Pd_{X\mid S=s}(x), 1\right].
	\end{align*}
	Now observe that in order for the protocol to output a value 
	$p\in \cS(x)$ (or $\cM(x), \cL(x)$, respectively),
	the value $J$ chosen in the protocol must be in the set $\cS'(x)$
	(or $\cM'(x), \cL'(x)$, respectively)	defined as follows:
	\begin{align*}
		\cS'(x)&:=\bigl\{j: 2^{-(j+1)\eps} \leq
									 \frac{\Pd_{X\mid S=s}(x)}{2^{(\Gsize/2-1)\eps}}
						\bigr\}, \\
		\cM'(x)&:=\bigl\{j: 
		2^{-j\eps} > \frac{\Pd_{X\mid S=s}(x)}{2^{(\Gsize/2-1)\eps}} 
		\land
		2^{-(j+1)\eps} < 2^{(\Gsize/2-1)\eps} \Pd_{X\mid S=s}(x) 
						\bigr\},\\
		\cL'(x)&:= \bigl\{j: 2^{-j\eps} \geq 2^{(\Gsize/2-1)\eps} \Pd_{X\mid S=s}(x)  
						\bigr\}.				
	\end{align*}	

	\begin{proof}[Proof of Lemma~\ref{lem:soundFixedS} (i)]
	We first calculate
	\begin{align}
		 \Pr[P\leq &f_s(X)|S=s] = \sum_{x \in\{0,1\}^n}\sum_{p\leq f_s(x)} \Pd_{XP\mid S=s}(x,p) \nonumber \\
		&=	
		 \sum_{x}
			\sum_{p \in \cS(x)} \Pd_{XP\mid S=s}(x,p) 
		 \leq 
			\sum_x \sum_{j \in \cS'(x)} \sum_{p \in \cA_j}\Pd_{XP\mid S=s}(x,p) \nonumber \\
		&\stackrel{\text{Lem.~\ref{lemma:qvsr}}}{\leq}
			\frac{2^{\eps}}{w(s)}
				\sum_{x}
				\sum_{j \in \cS'(x)}
		 	  \frac{r_s(x,j)}{2^{j\eps}}. \label{eqn:21}
	\end{align}
	To bound this, we find
	\begin{align}	
		\sum_{x}
				\sum_{j \in \cS'(x)}
		 	  \frac{r_s(x,j)}{2^{j\eps}}
		&=	\sum_{x}
				\sum_{k}
				\sum_{j \in \cS'(x) \cap \cI_k}
		 	  \frac{r_s(x,j)}{2^{j\eps}} 
		\nonumber \\
		&\leq\sum_{x}
				\sum_{k}
				\frac{1}{2^{\min(\cS'\cap \cI_k)\eps}}
				\sum_{j \in \cS'(x) \cap \cI_k} r_s(x,j) \nonumber\\
		&\stackrel{\text{Claim~\ref{claim:rsumtoone}}}{\leq}
			\sum_{x}
				\sum_{k}
				\frac{1}{2^{\min(\cS'\cap \cI_k)\eps}} 
				\cdot \left[\cS'\cap \cI_k \neq \emptyset\right]
				\nonumber\\
		&\leq 
			\sum_{x}
				\frac{\Pd_{X\mid S=s}(x)}{2^{(\Gsize/2-2)\eps}}
				\sum_{i=0}^{\infty}
				\frac{1}{2^{i\cdot\Gsize\cdot\eps}} \nonumber \\
		&\leq 
			\sum_{x}
				\frac{\Pd_{X\mid S=s}(x)}{2^{(\Gsize/2-2)\eps}}
				\sum_{i=0}^{\infty}
				\frac{1}{2^i}
		\leq \frac{2}{2^{(\Gsize/2-2)\eps}}.
	\end{align}
In the third inequality we used the definition of $\cS'$ and $\cS' \cap \cI_k \subseteq \cS'$, and that for any $k$ the distance between any two elements $j \in \cI_k$ and $j' \in \cI_{k+1}$
is at least $\Gsize$. In the fourth inequality we used 
that by definition of $\Gsize$, we have 
$2^{i\cdot \Gsize \cdot \eps} \geq 
 2^{i\cdot \Gsize' \cdot \eps}
 = (\frac{1}{\eps})^i \geq 2^i$ as $\eps < 1/2$. 
Thus the probability in (\ref{eqn:21}) is bounded by
$
\frac{2\cdot 2^{\eps}}{w(s) 2^{(\Gsize/2-2)\eps}} 
\leq 
\frac{2\cdot 2^{\eps}}{w(s) 2^{(\Gsize'/2-2)\eps}} 
=
\frac{2\eps 2^{3\eps}}{w(s)} \leq \frac{16 \eps}{w(s)}
$,
where we used the definition of $\Gsize'$ and 
$\Gsize' \leq \Gsize$.
\end{proof}	
	
\begin{proof}[Proof of Lemma~\ref{lem:soundFixedS} (ii)]
Fix any $x\in \{0,1\}^n$. 
We get
\begin{align}
	\sum_{p > f_s(x)} \frac{\Pd_{XP\mid S=s}(x,p)}{p} = \sum_{p\in \cM(x)}\frac{\Pd_{XP\mid S=s}(x,p)}{p}
		+ \sum_{p\in \cL(x)}\frac{\Pd_{XP\mid S=s}(x,p)}{p}  \label{eqn:22}
\end{align}

We bound the two terms separately, and first bound the sum over $p\in \cL(x)$. By definition of $\cL$, $p\geq 2^{(\Gsize/2 - 1)\eps}\Pd_{X\mid S=s}(x)$, and thus we find that 
\begin{align}
	\sum_{p\in \cL(x)}\frac{\Pd_{XP\mid S=s}(x,p)}{p} 
		&\leq \frac{1}{2^{(\Gsize/2 - 1)\eps}\Pd_{X\mid S=s}(x)} 
			\underbrace{\sum_{p\in \cL(x)}\Pd_{XP\mid S=s}(x,p)}_{\leq \Pd_{X\mid S=s}(x)} \nonumber \\
		 &\leq \frac{1}{2^{(\Gsize/2 - 1)\eps}}
		 =\eps 2^{\eps} \leq 2\eps. \label{eqn:23}
\end{align}

We proceed to bound the sum over $\cM(x)$. We find
\begin{align}
	&\sum_{p\in \cM(x)}\frac{\Pd_{XP\mid S=s}(x,p)}{p}
	  \leq \sum_{j\in \cM'(x)}\sum_{p \in \cA_j}\frac{\Pd_{XP\mid S=s}(x,p)}{p} \nonumber \\
		 &\leq \sum_{j\in \cM'(x)}2^{(j+1)\eps} \sum_{p \in \cA_j}\Pd_{XP\mid S=s}(x,p)
	  \stackrel{\text{Lem.~\ref{lemma:qvsr}}}{\leq} 
		  \sum_{j\in \cM'(x)}
			\frac{2^{(j+1)\eps} r_s(x,j) 2^{\eps}}
					 {w(s)2^{j\eps}}\nonumber \\
		&= \frac{2^{2\eps}}
					 {w(s)}
			\sum_{j\in \cM'(x)} r_s(x,j)
		\leq \frac{2^{2\eps}}{w(s)}
			\label{eqn:25}
\end{align}
The second inequality holds because $p\in \cA_j$ implies 
$p \geq 2^{-(j+1)\eps}$. 
The fourth inequality follows by
Claim~\ref{claim:rsumtoone} once we observe that
by definition of $\cM'(x)$ we have 
\begin{align*}  
	\cM'(x) \subseteq \Bigl\{ \Bigl\lceil  
												\frac{1}{\eps}\log_2\bigr(\frac{1}{\Pd_{X\mid S=s}(x)}\bigl) 
												-&\frac{\Gsize}{2}
										 \Bigr\rceil
										 , \ldots, \\ 
										 &\Bigl\lfloor 
												\frac{1}{\eps}\log_2\bigl(\frac{1}{\Pd_{X\mid S=s}(x)}\bigr) 
												+\frac{\Gsize}{2}-1
										 \Bigr\rfloor
						 \Bigr\}, 
\end{align*}
which implies $|\cM'(x)| \leq \Gsize$ and thus there exists a $k$
such that for all $j\in \cM'(x)$ with $r_s(x,j) >0$ we have $j\in \cI_k$.
Plugging (\ref{eqn:23}) and (\ref{eqn:25}) into~(\ref{eqn:22}) concludes the proof.
\end{proof}

\subsubsection{Proof of the main theorem}
\label{sec:ProofOfMainSamplingTheorem}
We prove the main theorem given Lemmas~\ref{lem:runningTime}, \ref{lem:complFixedS} and
\ref{lem:soundFixedS}. The running time is given by Lemma~\ref{lem:runningTime}, and
so it only remains to prove completeness and soundness.

\paragraph{Proof of completeness.}
We will use the following claim:
\begin{claim}
	\label{claim:setS}
	Fix any prover and let $\cN_s := \bigcup_{i \in \cK} \cI_i(s)$.
	Then the following holds:
	\begin{enumerate}[(i)]
		\item The sets $\{\overline{\cN}_s\}_{s\in \cS}$ 
		form a partition of $(t)$,
		\item $|\cS| = \frac{\Isize}{\Gsize}$, and $\frac{1}{2\delta} \leq |\cS| \leq \frac{1}{\delta}$,
		\item $\sum_{s \in \cS} w(s) \in [1-2^{-n},1]\cdot (|\cS|-1)$,
		\item $\frac{|\cS|}{\sum_{s\in \cS}w(s)} \leq 1+16\delta$.
	\end{enumerate}
\end{claim}

\begin{proof}
	Part (i) and (ii) follow by definition.
  To see (iii), we calculate
	\begin{align*}
	\sum_{s \in \cS} w(s) &= \sum_{s\in \cS} \sum_{i\in \cK} \sum_{j\in \cI_i(s)} h_j = \sum_{s\in \cS}\sum_{j\in \cN_s} h_j \\
	&= (|\cS|-1)\sum_{i \in (t)} h_j \in [1-2^{-n},1]\cdot (|\cS|-1). 
	\end{align*}
	where the third inequality follows since part (i) implies that each
	$j \in (t)$ occurs in exactly $(|S|-1)$ many sets $\cN_s$. 
	
	To see (iv), we calculate
	\begin{align*}
	\frac{|\cS|}{\sum_{s\in \cS}w(s)}
		&\stackrel{\text{(iii)}}{\leq} \frac{|\cS|}{(1-2^{-n})(|S|-1)}
		=\frac{1}{1-2^{-n}-\frac{1}{|S|}+\frac{2^{-n}}{|S|}}
		\leq \frac{1}{1-2^{-n}-\frac{1}{|S|}} \\
		&\leq \frac{1}{1-\frac{2}{|S|}} 
		\stackrel{\text{(ii)}}{\leq}
			\frac{1}{1-4\delta} \leq 1+16\delta.
	\end{align*}
	For the third inequality we used that $|\cS| < 2^n$, which is implied by $|\cS|\leq \frac{1}{\delta}$ from (ii) and using
	$\bigl(\frac{1}{\eps}\bigr)^{1/\delta}\leq 2^{n/50}$.
\end{proof}

\begin{proof}[Proof of completeness]
	Suppose the prover is honest. We prove completeness for the set 
	$\cM$ as provided by Lemma~\ref{lem:complFixedS}. This lemma also
	gives that the protocol only outputs pairs of the form $(x,\Pd(x))$
	where $x\in \cM$. 
	Fix any $x \in \cM$	and let $p:=\Pd(x)$.
	Since the sets $\{\overline{\cM}_s \cap \cM \}_{s\in \cS}$ form
	a partition of $\cM$, 
	we have that there is exactly one $s^* \in \cS$
	such that $x \notin \cM_{s^*}$.
	We find
	\begin{align*}
		\Pd_{XP}(x,p)
		&= \sum_{s\in \cS} \Pd_S(s) \Pd_{XP\mid S=s}(x,p)
		= \sum_{s\in \cS \setminus \{s^*\}} \Pd_S(s) \Pd_{XP\mid S=s}(x,p)\\
		&= \sum_{s\in \cS \setminus \{s^*\}}	
			\frac{w(s)}{\sum_{s'\in S}w(s')} \cdot \Pd_{XP\mid S=s}(x,p) \\
		&\stackrel{\text{Lem.~\ref{lem:complFixedS} (i)}}{\in }
		(1\pm 132\eps) \sum_{s\in \cS \setminus \{s^*\}}
					\frac{w(s)}{\sum_{s'\in S}w(s')} \cdot \frac{\Pd(x)}{w(s)}\\
		&= (1\pm 132\eps) \Pd(x) \frac{|\cS|-1}
				{\sum_{s'\in S}w(s')}.
	\end{align*}
	Applying item (iii) of Claim~\ref{claim:setS} gives
	\begin{align*}
		(1-\eps')\Pd(x) &\leq (1-132\eps) \Pd(x) 
		\leq \Pd_{XP}(x,p)
		\leq \frac{1+132\eps}{1-2^{-n}}\Pd(x)\\
		&\leq (1+136\eps)\Pd(x) \leq (1+\eps')\Pd(x),
	\end{align*}
	where we used that $\eps \geq 2^{-n}$, which follows from the assumption $\left(\frac{1}{\eps}\right)^{1/\delta}\leq 2^{n/50}$.
\end{proof}

\paragraph{Proof of soundness.}
To prove soundness, we first define the event $\Bad$. 
	For $s\in \cS$ let 
	$f_s$ be the function given by Lemma~\ref{lem:soundFixedS}.
	The event $\Bad$ occurs if 
	$(S,X,P)=(s,x,p)$ and $p \leq f_s(x)$.

\begin{proof}[Proof of soundness (i)]
	Part (i) of Lemma~\ref{lem:soundFixedS}	states that
	$\Pr[\Bad|S=s] \leq \frac{16\eps}{w(s)}$.
	Thus we find
	\begin{align*}
		\Pr[\Bad] 
		&= \sum_{s\in \cS}\Pd_S(s)\Pr[\Bad|S=s] 
		= \sum_{s\in \cS}\frac{w(s)}{\sum_{s'\in S}w(s')}\Pr[\Bad|S=s] \\
		&\leq \sum_{s\in \cS}\frac{w(s)}{\sum_{s'\in S}w(s')} 
					\frac{16\eps}{w(s)} 
		= \frac{16\eps|S|}{\sum_{s'\in S}w(s')} 
		\leq 16\eps (1+16\delta) \leq 300\eps \leq \eps'.
	\end{align*}
	The second inequality holds by item (iv) of Claim~\ref{claim:setS}, and 
	the last inequality holds since $\delta \leq 1$.
\end{proof}

\begin{proof}[Proof of soundness (ii)]
	Fix any $x \in \{0,1\}^n$. We find
\begin{align*}
		\sum_{p}&\frac{\Pr[(X,P)=(x,p)\land\lnot\Bad]}
													 {p}\\
		&= \sum_{p} \sum_{s\in \cS}
			\frac{\Pd_S(s)\Pr[(X,P)=(x,p)\land \lnot \Bad|S=s]}{p}\\
		&= \sum_{s\in \cS} \Pd_S(s) \sum_{p>f_s(x)}
				\frac{\Pd_{XP\mid S=s}(x,p)}{p} 
		\stackrel{\text{Lem.~\ref{lem:soundFixedS} (ii)}}{\leq} \sum_{s\in \cS} \frac{w(s)}{\sum_{s'\in \cS}w(s')} 
					\frac{(1+6\eps)}{w(s)}\\
		&= \frac{|S|(1+6\eps)}{\sum_{s'\in \cS}w(s')}
		\stackrel{\text{Claim~\ref{claim:setS} (iv)}}{\leq} 
			(1+16\delta)\left(1+6\eps\right).
\end{align*}
Thus, we conclude that
\begin{align*}
	&\sum_{p}\frac{\Pr[(X,P)=(x,p)|\lnot \Bad]}{p}\\
		&\qquad =
			\frac{1}{\Pr[\lnot \Bad]}
			\sum_{p}\frac{\Pr[(X,P)=(x,p)\land \lnot \Bad]}
					 {p}\\
		&\qquad \leq 
		  \frac{1}{1-300\eps}
				(1+16\delta)\left(1+6\eps\right) 
		\leq (1+600\eps)(1+16\delta)(1+6\eps) \\
		&\qquad\leq 1+9000\eps + 16\delta \leq 1+\eps'+\delta'. 
			\qedhere
\end{align*}
\end{proof}

\subsubsection{The verifier's running time}
\label{sec:VerifiersRunningTime}

\begin{proof}[Proof of Lemma~\ref{lem:runningTime}]
	It is sufficient to show that $t$ and the number of
	elements in the sets $\cX_i$ sent by the prover are
	polynomial. Clearly $t = \poly(n/\eps)$. 
	
	We now consider the sets $\cX_i$. Fix any $k\in \cK$
	and consider the corresponding value $m$. We first consider
	the case where $m=0$. By definition of $m$, this implies
	$g \geq \log_2\left(\sum_{i\in \cI_k}2^{i\eps}h_i\right)$. 
	Thus we find
	\begin{align*} 
		 \sum_{i \in \cI'_k}|\cX_i| 
		 &\leq \sum_{i\in \cI_k'}2^{(i+1)\eps}h_i
	   \leq 2^{\eps}\sum_{i\in \cI_k}2^{i\eps}h_i
		 \leq 2^{\eps}2^g
		 \leq 2\cdot 2^g 
		 \leq 4\cdot 2^{\SampGap} \\
		 &= 4 \cdot \frac{t \cdot \Isize' 
										 \cdot 2^{2\cdot \Isize'\eps}}
										{\eps^4}
		 = \frac{8 t\log_2\left(\frac{1}{\eps}\right) 
			 \left(\frac{1}{\eps}\right)^{4/\delta}}{\delta\eps^5}\\
		 &= \poly\left(n\cdot \left(
					\frac{1}{\eps}\right)^{1/\delta}\right).
	\end{align*}
	In case $m>0$ we find
	\begin{align*}
		\sum_{i \in \cI'_k}|\cX_i| 
		&\leq \sum_{i \in \cI'_k} 2^{\eps}
						\frac{2^g2^{(i+1)\eps}h_i}
								 {\sum_{j\in \cI_k}2^{j\eps}h_j}
		= 2^{\eps}2^g2^{\eps}
			\underbrace{\sum_{i \in \cI'_k}\frac{2^{i\eps}h_i}
								 {\sum_{j\in \cI_k}2^{j\eps}h_j}}_{\leq 1}
		\leq 4\cdot 2^g \\
		&= \poly\left(n\cdot \left(\frac{1}{\eps}\right)^{1/\delta}\right),
	\end{align*}
	where the last equality follows as in the first case above.
\end{proof}

\subsubsection{Bounding the probability of the event $p\in \cA_j$}
\label{sec:BoundingProbOfpInAj}

\begin{proof}[Proof of Lemma~\ref{lemma:qvsr}]
	Fix $s\in \cS$. By definition of the protocol, for any $x\in \{0,1\}^n$ and
	$j \notin \bigcup_{\ell \in \cK} \cI'_{\ell}(s)$, 
	we have $\sum_{p\in \cA_j}\Pd_{XP\mid S=s}(x,p)=0=r_s(x,j)$, which implies 
	the lemma. 
	For any 
	$x \in \{0,1\}^n$ and	$j \in \bigcup_{\ell \in \cK} \cI'_{\ell}(s)$,
	we find the following. 
	\begin{align}
		&\sum_{p\in \cA_j}\Pd_{XP\mid S=s}(x,p) 
			= \Pr_{X,P}[X=x \land P \in \cA_j | S=s]
			= \Pr_{X,J}[(X,J)=(x,j)|S=s] \label{eqn:26} \\
		&= \sum_{k} \Pr_{F,J,X}\Bigl[(X,J)=(x,j) \Big| 
			(K,S)=(k,s)\Bigr] \cdot \Pr_K[K=k|S=s] \nonumber \\
		&= \Pr_{F,J,X}\Bigl[(X,J)=(x,j)\Big| 
			(K,S)=(k_s(j),s)\Bigr] \cdot \Pr_K[K=k_s(j)|S=s]\nonumber \\
		& = \frac{w_{k_s(j)}}{w(s)}
			\cdot \Pr_{F,J,X}\Bigl[F(x)=0^m\Big| 
			(K,S)=(k_s(j),s)\Bigr] \nonumber\\
		& \qquad \qquad \cdot \Pr_{F,J,X}\Bigl[(X,J)=(x,j)\Big| 
			(K,F(x),S)=(k_s(j),0^m,s)\Bigr] \nonumber \\
		& = \frac{w_{k_s(j)}}{2^m w(s)}
				\cdot 
				\Pr_{F,J,X}\Bigl[(X,J)=(x,j) \Big| 
			(K,F(x),S)=(k_s(j),0^m,s)\Bigr]\nonumber  \\
		&= \frac{w_{k_s(j)}}{2^m w(s)} 
		\cdot \Pr_J[J=j|(K,F(x),S)=(k_s(j),0^m,s)] \nonumber  \\
		&\qquad \cdot \Pr_{F,X}\Bigl[
			X=x 
			\Big| (K,F(x),S,J)=(k_s(j),0^m,s,j)\Bigr]\nonumber  \\
		&= \frac{w_{k_s(j)}}{2^m w(s)} 
			\cdot \Pr_J[J=j|(K,F(x),S)=(k_s(j),0^m,s)] \label{eqn:18}\\
			&\qquad \cdot \Pr_{X}\Bigl[
					X=x 
			\Big|(K,F(x),J,S)=(k_s(j),0^m,j,s) \land
			  (a)-(c) \land x \in \cX_{j}\Bigr] \label{eqn:19}\\
		&\qquad \cdot	\Pr_{F}\Bigl[
					(a)-(c) \land
					x \in \cX_{j} 
			\Big|(K,F(x),J,S)=(k_s(j),0^m,j,s)\Bigr] \label{eqn:20}
	\end{align}			
	The second inequality in (\ref{eqn:26}) follows because
	$P\in \cA_j$ if and only if $J=j$. 
	The probability in (\ref{eqn:20}) equals $r_s(x, j)$
	since the prover's
	choice of the sets $\cX_i$ is independent of $J$. 
	The probability in (\ref{eqn:18}) equals
	$\frac{h_{j}}{\sum_{i\in \cI_{k_s(j)}}h_i}$.
	The probability in (\ref{eqn:19}) equals $\frac{1}{|\cX_{j}|}$, and to
	bound it we distinguish two cases:
	If $m=0$, condition (b) implies
	\begin{align*}
			\frac{1}
					 {2^{(j+1)\eps} h_{j}}
			\leq \frac{1}{|\cX_{j}|} \leq
			\frac{1}
				 {2^{j\eps} h_{j}},
	\end{align*}
	and if $m>0$, condition (b) implies
	\begin{align*}
			2^{-\eps}\frac{\sum_{i\in \cI_{k_s(j)}}2^{i\eps}h_i}
					 {2^{g} 2^{(j+1)\eps} h_{j}}
			\leq \frac{1}{|\cX_{j}|} \leq
			2^{\eps}\frac{\sum_{i\in \cI_{k_s(j)}}2^{i\eps}h_i}
				 {2^{g} 2^{j\eps} h_{j}}.
	\end{align*}
	Using the definition of $m$ and $w_{k_s(j)}$, this gives the claim.
\end{proof}

\subsubsection{Proof of completeness: the details}
\label{sec:SamplingProofCompletenessDetails}

\begin{claim}
	\label{claim:sumwilarge}
	$\sum_{i\in (t)}h_i \geq 1-2^{-n}$.
\end{claim}
\begin{proof}
	Note that the only elements not in $\bigcup_{i\in (t)}\cB_i$
	are the elements $\overline{\cB}:=\{x: \Pd(x) \leq 
	2^{-(t+1)\eps}\}$. 
	By definition of $t$ we have $2^{-(t+1)\eps}\leq 
	2^{-\lceil 2n/\eps\rceil \eps} \leq 2^{-2n}$. 
	The claim follows by observing that
	$\sum_{x\in \overline{\cB}}\Pd(x) \leq |\{0,1\}^n|\cdot 2^{-2n} = 2^{-n}$.
\end{proof}

\begin{claim}
	\label{claim:propertiesOfM}
	We have
	\begin{enumerate}[(i)]
		\item $\Pd(\cM) \geq 1-2\eps$
		\item The sets $\{\overline{\cM}_s \cap \cM\}_{s\in \cS}$ 
		form a partition of $\cM$.
	\end{enumerate}
\end{claim}
\begin{proof}
	To see (i), we note that
	$\Pd(\cM) 
		\geq \sum_{i\in (t)}h_i-t\cdot \frac{\eps}{2t}
		\geq 1-2^{-n}-\eps \geq 1-2\eps$,
	where the first inequality follows by definition of $\cN$, 
	and the second inequality holds by 
	Claim~\ref{claim:sumwilarge}.
	Part (ii)	follows by definition.
\end{proof}

\begin{proof}[Proof of Claim~\ref{claim:BareBig}]
		Using the definition of $m$ and 
		$h_i \geq \frac{\eps}{2t}$, we find
		\begin{align*}
			\frac{|\cB_i|}{2^m} &= 
			\frac{|\cB_i|2^g}
					 {\sum_{{\ell}\in \cI_{k}}2^{\ell\eps}h_{\ell}}			
			\stackrel{\text{Claim~\ref{claim:sizeBi}}}{\geq}	
			\frac{2^g2^{i\eps}h_i}
					 {\sum_{{\ell}\in \cI_{k}}2^{\ell\eps}h_{\ell}}
			\geq
				\frac{2^g2^{i\eps}h_i}{\sum_{\ell \in \cI_{k}} 
				2^{\max(\cI_{k})\eps}h_{\ell}} \\
			&\geq
				\frac{2^g2^{i\eps}h_i}{\sum_{\ell \in \cI_{k}} 
				2^{(i+\Isize)\eps}h_{\ell}}
			\geq
				\frac{2^g2^{i\eps}h_i}
				{2^{(i+\Isize)\eps}}	 
			= \frac{2^gh_i}	
							{2^{\Isize\cdot \eps}}		\\	 
			&\geq
				 \frac{2^{\SampGap} \eps}
				 		  {2 t2^{\Isize\cdot \eps}} 
		  \stackrel{\text{Claim~\ref{claim:ParamProperties} (iii)}}{\geq }
				\frac{\Isize}{4\eps^3}.\qedhere
		\end{align*}
	\end{proof}
	
\begin{proof}[Proof of Lemma~\ref{lem:RejectWithSmallProb}]
	By Claim~\ref{claim:sumwilarge}, the verifier does not reject after
	receiving the prover's first message.

	Fix $s\in \cS$. We first bound the probability that the verifier rejects
	because $j \notin \cI_k'$
	given $S=s$. 	
	First note that
	\begin{align*}
	\Pr[J=j|S=s] 
		&= \Pr[J=j \land K=k_s(j)|S=s] \\
		&= \Pr[J=j |(S,K)=(s,k_s(j))]
			\cdot \Pr[K=k_s(j)|S=s] \\
		&\leq \frac{w_{k_s(j)}}{w(s)}\cdot \frac{h_j}{w_{k_s(j)}} 
		= \frac{h_j}{w(s)},
	\end{align*}
	where we write an inequality above because the verifier 
	may reject in (a)-(c) before choosing $j$.
	Thus, using the definition of $\cN$, we find
	\begin{align*}
		\Pr[J \notin \cN] 
		&=		\sum_{j\notin \cN} \Pr[J=j|S=s]
		\leq	\sum_{j\notin \cN} \frac{h_j}{w(s)} 
		\leq \sum_{j\notin \cN} \frac{\eps}{2 \cdot t \cdot w(s)}\\
		&\leq t \cdot \frac{\eps}{2 \cdot t \cdot w(s)}
		=    \frac{\eps}{2 \cdot w(s)}.
	\end{align*}	
	
  We bound the probability that the verifier rejects
	in (a)-(c) given $S=s$.	
	Since the prover is honest, (a) and (c) always hold. It remains
	to bound the probability that the verifier rejects in (b). If 
	$m=0$, the verifier does not reject by Claim~\ref{claim:sizeBi}. 
	So assume $m>0$ and suppose
	$K=k$ for some $k\in \cK$. Then for any 
	$i \in \cI_k'$ we find
	\begin{align*}
		\Pr_F\Bigl[ |\cX_{i}| &\notin
										\Bigl[2^{-\eps} 
											\frac{2^{g} 2^{i\eps} h_i}
													 {\sum_{{\ell}\in \cI_{k}}
													2^{\ell\eps}h_{\ell}},
										2^{\eps}  
											\frac{2^{g} 2^{(i+1)\eps} h_i}
													 {\sum_{{\ell}\in
													  \cI_{k}}2^{\ell\eps}h_{\ell}}\Bigr]
										\Bigr] 
		\leq 	
			\Pr_F \Bigl[ |\cX_{i}| \notin
										\Bigl[2^{-\eps}  
											\frac{|\cB_i|}
													 {2^m},
										2^{\eps} 
											\frac{|\cB_i|}
													 {2^m}
										\Bigr] \Bigr] \\
		&\leq
		\frac{2^m}{(\eps/2)^2|\cB_i|}
		\leq \frac{16\eps}{\Isize},
	\end{align*}	
	where the first inequality holds by definition of $m$ and 
	Claim~\ref{claim:sizeBi}, the second inequality is an application 
	of part (i) of the hash mixing lemma (Lemma~\ref{lemma:hashMixingLemma}) using 
	$(1\pm \eps/2)\subseteq 2^{\pm \eps}$, and the third inequality
	follows using Claim~\ref{claim:BareBig}. Noting that 
	$|\cI_k'|\leq \Isize$ and by the union bound, the
	probability that there exists some $i\in \cI_k'$ such that $|\cX_i|$
	is not in the desired interval is at most $16\eps$. 
	
	Thus, we conclude that the probability that the verifier rejects is
	bounded by
	\begin{align*}
		&\sum_{s\in \cS}\Pd_S(s)\cdot 
				 \left(\frac{\eps}{2 \cdot w(s)}+16\eps \right)
		= 16\eps + \sum_{s\in \cS}\frac{w(s)}{\sum_{s'\in \cS }w(s')}\cdot \frac{\eps}{2 \cdot w(s)} \\
		&\qquad = 16\eps + 
			\frac{\eps}{2}\frac{|\cS|}{\sum_{s'\in \cS }w(s')}
		\stackrel{\text{Claim~\ref{claim:setS} (iv)}}{\leq} 16\eps + \frac{\eps}{2}(1+16\delta) \leq 25\eps.\qedhere
	\end{align*} 
\end{proof}
				
\begin{proof}[Proof of Lemma~\ref{lem:complFixedS}]
	For a fixed $k$, let $m_s(k)$ be the value of $m$ in the protocol 
	given $K=k$ and $S=s$.
	By definition of the protocol, we have $\Pd_{XP\mid S=s}(x,p)=0$ if $p\neq \Pd(x)$ or $x\notin \cM \cap \cM_s$.
	Fix any $s\in \cS$ and $x \in \cM \cap \cM_s$. Then for some $j$ we have
	$x \in \cB_{j}$, $j \in \cI'_{k_s(j)}$ and $h_j\geq \frac{\eps}{2t}$.
	Since the prover is honest, we have
	$	r_s(x,j) = \Pr_{F}[(b)|(K,F(x),S)=(k_s(j),0^m,s)]$.
	In case $m_s(k_s(j))=0$ we find
	\begin{align*}
			r_s(x,j) = \Bigl[
												\forall i\in \cI'_{k_s(j)}: 
											2^{i\eps} h_i
											\leq |\cX_{i}| \leq
											2^{(i+1)\eps} h_i 
										\Bigr]  = 1,
	\end{align*}
	since the event $F(x)=0^0$ is always true, the honest prover defines 
	$\cX_i := \cB_i$, and we used Claim~\ref{claim:sizeBi}.

	We now consider the case $m_s(k_s(j))>0$. We find
	\begin{align}
					&r_s(x,j)
					= \Pr_{F}[(b)|(K,F(x),S)=(k_s(j),0^m,s)]  \nonumber\\
					&= \Pr_{F}\Bigl[
												\forall i\in \cI'_{k_s(j)}: 
										2^{-\eps}
											\frac{2^{g} 2^{i\eps} h_i}
													 {\sum_{\ell \in \cI_{k_s(j)}}
													 2^{\ell\eps}h_{\ell}}
										\leq |\cX_{i}| \leq
										2^{\eps}
											\frac{2^{g} 2^{(i+1)\eps} h_i}
													 {\sum_{\ell \in \cI_{k_s(j)}}
													 2^{\ell\eps}h_{\ell}} \nonumber\\
										&\qquad\qquad\qquad\qquad\qquad\qquad
										\Big|(F(x),S)=(0^m,s)\Bigr] \nonumber \\
					&\geq
							1- \sum_{i\in \cI'_{k_s(j)}}
									\Pr_{F}\Bigl[ |\cX_{i}| \notin
										\Bigl[2^{-\eps}
											\frac{2^{g} 2^{i\eps} h_i}
													 {\sum_{{\ell}\in \cI_{k_s(j)}}
													2^{\ell\eps}h_{\ell}},
										2^{\eps}
											\frac{2^{g} 2^{(i+1)\eps} h_i}
													 {\sum_{{\ell}\in
													  \cI_{k_s(j)}}2^{\ell\eps}h_{\ell}}\Bigr]\nonumber\\
										&\qquad\qquad\qquad\qquad\qquad\qquad
										\Big|(F(x),S)=(0^m,s)\Bigr]\nonumber \\
					&\geq 
							1- \sum_{i\in \cI'_{k_s(j)}}
							\Pr_{F} 
					 \Bigl[|\cX_i|\notin\Bigl[
					 				2^{-\eps}\frac{|\cB_i|}{2^m},
					 				2^{\eps}\frac{|\cB_i|}{2^m}\Bigr]
					 \Big|(F(x),S)=(0^m,s)\Bigr] 
					\label{eqn:15} 
	\end{align}
	where the first inequality follows by the union bound, and the second
	inequality holds by definition of $m$ and 
	Claim~\ref{claim:sizeBi}. To bound the
	probabilities in the sum in (\ref{eqn:15}), we distinguish two cases. 
	The first case is $i\neq j$. By assumption we have $x\in \cB_{j}$ and
	thus $x \notin \cB_i$. We find
	\begin{align}
		&\Pr_{F} 
					 \Bigl[|\cX_i|\notin\Bigl[
					 				2^{-\eps}\frac{|\cB_i|}{2^m},
					 				2^{\eps}\frac{|\cB_i|}{2^m}\Bigr]
					 \Big|(F(x),S)=(0^m,s)\Bigr] \nonumber \\
		&\qquad\qquad\leq
		\frac{2^m}{(\eps/2)^2|\cB_i|}
		\leq \frac{16\eps}{\Isize},
		\label{eqn:27}
	\end{align}	
	where the first inequality is an application 
	of part (ii) of the hash mixing lemma (Lemma~\ref{lemma:hashMixingLemma}) using $(1\pm\eps/2)\in
	2^{\pm\eps}$, and the second inequality
	follows using Claim~\ref{claim:BareBig}.
	
	The second case is $i=j$. Then by assumption we have $x \in \cB_i$. 
	We find
	\begin{align}
		\Pr_{F} 
					 &\Bigl[|\cX_i|\notin\Bigl[
					 				2^{-\eps}\frac{|\cB_i|}{2^m},
					 				2^{\eps}\frac{|\cB_i|}{2^m}\Bigr]
					 \Big|(F(x),S)=(0^m,s)\Bigr] \nonumber \\
		&\leq	 \Pr_{F} 
					 \Bigl[|\cX_i|\notin\Bigl[
					 				1+\left(1-\frac{\eps}{4}\right)\frac{|\cB_i|-1}{2^m},
					 				1+\left(1+\frac{\eps}{4}\right)\frac{|\cB_i|-1}{2^m}\Bigr] \nonumber\\
					 &\qquad\qquad\qquad\qquad\qquad\qquad\qquad\qquad\qquad
					\Big|(F(x),S)=(0^m,s)\Bigr] \nonumber\\
		&\leq
		\frac{2\cdot 2^m}{(\eps/4)^2|\cB_i|}
		\leq
			\frac{2^7 \eps}{\Isize}.
			\label{eqn:28}
	\end{align}
	The second inequality is an application of the hash mixing lemma (Lemma~\ref{lemma:hashMixingLemma}) part 
	(iii) and using that $|\cB_i|\geq 2$ by Claim~\ref{claim:BareBig} combined with $\eps < 1/2$, 
	and the first inequality follows because on the one hand, we have
	\begin{align*}
		1+\left(1-\frac{\eps}{4}\right)\frac{|\cB_i|-1}{2^m} &\geq
		\left(1-\frac{\eps}{4}\right)(1+\frac{|\cB_i|-1}{2^m}) = \left(1-\frac{\eps}{4}\right)
		\frac{2^m-1+|\cB_i|}{2^m} \\
		&\geq \left(1-\frac{\eps}{4}\right)\frac{|\cB_i|}{2^m} 
		\geq (1-\frac{\eps}{2})\frac{|\cB_i|}{2^m}
		\geq 2^{-\eps}\frac{|\cB_i|}{2^m},
	\end{align*}
	where we used $2^m \geq 1$, and on the other hand it holds that
	$
		1+\left(1+\frac{\eps}{4}\right)\frac{|\cB_i|-1}{2^m}
		\leq 2^{\eps}\frac{|\cB_i|}{2^m}$.
	To see this, we calculate
	\begin{align*}	
		2^{\eps}&\frac{|\cB_i|}{2^m}-
		\Bigl(1+\left(1+\frac{\eps}{4}\right)\frac{|\cB_i|-1}{2^m}\Bigr) 
		\geq 2^{\eps}\frac{|\cB_i|}{2^m}-
		\left(1+\frac{\eps}{4}\right)\Bigl(1+\frac{|\cB_i|-1}{2^m}\Bigr) \\
		&\geq 
		2^{\eps}\frac{|\cB_i|}{2^m}
			-\left(1+\frac{\eps}{4}\right)\Bigl(1+\frac{|\cB_i|}{2^m}\Bigr) 
		= \frac{\eps}{4}\Bigl(\frac{|\cB_i|}{2^m} - 1\Bigr)-1 \\
		&\stackrel{\text{Claim~\ref{claim:BareBig}}}{\geq} \frac{\eps}{4}\Bigl(\frac{\Isize}{4\eps^3}-1\Bigr) - 1 
		\geq 0,
	\end{align*}
	where the last inequality follows since $\eps < 1/10$ and
	$\Isize \geq 1$. 
	
	Finally, plugging (\ref{eqn:27}) and (\ref{eqn:28}) into (\ref{eqn:15}), we find
	\begin{align*}
		r_s(x,j) \geq 1 - \Bigl((\Isize-1)
					\frac{16 \eps}{ \Isize}
					+ \frac{2^7 \eps}{\Isize} 
						\Bigr) 
		\geq 1-2^7\eps.
	\end{align*}		
	
	We showed that in both cases $m_s(k_s(j))=0$ and $m_s(k_s(j))>0$, 
	we have $1-2^7\eps \leq r_s(x, j) \leq 1$.
	Plugging this into Lemma~\ref{lemma:qvsr} and 
	since the honest prover always sends $p=\Pd(x)$, we find
	\begin{align*}
		q\bigl(x,\Pd(x)\bigr) \in 
			&\left[
				2^{-2\eps}
				(1-2^7\eps) \frac{2^{-j\eps}}{w(s)},
				2^{\eps}
				\frac{2^{-j\eps}}{w(s)}			
			\right] 
	\subseteq 
			\left[
				2^{-2\eps}
				(1-2^7\eps) \frac{\Pd(x)}{w(s)},
				2^{\eps}
				\frac{2^{\eps}\Pd(x)}{w(s)}	
			\right]\\
	&\subseteq (1\pm 132\eps)\frac{\Pd(x)}{w(s)},
	\end{align*}
	where we used that
	$x\in \cB_j$ implies $2^{-j\eps}\in [\Pd(x), 2^{\eps}\Pd(x)]
	$ for the first inclusion, and $2^{\pm 2\eps} \subseteq 
	(1\pm 4\eps)$ for the second inclusion.
\end{proof}

\subsubsection{Proof of soundness: the details}
\label{sec:SamplingProofSoundnessDetails}
\begin{proof}[Proof of Claim~\ref{claim:rsumtoone}]
	Letting $\cI_k := \cI_k(s)$ and $\cI'_k := \cI'_k(s)$ we find
 	\begin{align*}
 		\sum_{j \in \cI_k}
							 r_s(x,j)
		&= \sum_{j \in \cI'_k}
							 r_s(x,j) \\
		&=\sum_{j \in \cI'_k}
							 \Pr_{F}\Bigl[(a)-(c) \land	x \in \cX_{j}
							 \Big|(K,F(x),S)=(k,0^m,s)\Bigr]  \\
		&= \sum_{j \in \cI'_k}
					\sum_{f:f(x)=0^m}
							 \Pr_F[F=h|(K,F(x),S)=(k,0^m,s)]	\\
		&\qquad \cdot			\Bigl[(a)-(c) \land	x \in \cX_{j} 
							\Big|(K,F(x),F,S)=(k,0^m,h,s)\Bigr]  \\
		&= \sum_{f:f(x)=0^m}
							 \Pr_F[F=h|(K,F(x),S)=(k,0^m,s)]\\
		&\qquad \cdot
							\sum_{j \in \cI'_k}
							\Bigl[(a)-(c) \land	x \in \cX_{j}
							\Big|(K,F(x),F,S)=(k,0^m,h,s)\Bigr]\\
		&\leq \sum_{f:f(x)=0^m}
							 \Pr_F[F=h|(K,F(x),S)=(k,0^m,s)]
					\cdot 1 =  1,
 	\end{align*}
 	where the last inequality holds since the prover either makes the 
	verifier reject in (c), or sends sets
 	$\cX_i$ that are pairwise disjoint.
\end{proof}

\section{Private Coins Versus Public Coins in Interactive Proofs}
\label{chap:PrivCoinsPubCoins}

\subsection{Theorem Statement and Discussion}
\label{sec:PubCoinThmAndDiscussion}
An interactive proof has completeness $c$ and soundness $s$ if the verifier accepts with probability at least $c$ in case the prover's claim is true, and at most $s$ if the claim is false. Recall from Definition~\ref{def:InteractiveProofs} that $\IP$ denotes the set of languages that admit a $k(n)$-round private-coin interactive proof (we write $\textsf{rounds}=k(n)$) with message length $m(n)$ ($\textsf{msg size}=m(n)$), $\ell(n)$ many random coins ($\textsf{coins} = \ell(n)$), completeness $c(n)$ ($\textsf{compl}\geq c(n)$) and soundness $s(n)$ ($\textsf{sound}\leq s(n)$), where the verifier runs in time $t(n)$ ($\textsf{time}= t(n)$), where $n$ denotes the input length. Also recall that for the public-coin case, the set $\AM$ is defined analogously. 

We state our theorem for the case where the public-coin verifier calls the private-coin verifier exactly once. This highlights the overhead in the verifier's running time, and allows to compare our transformation to the one of \cite{GS86} in a natural way. We remark that before applying the transformation as given in the theorem, one can repeat the private-coin protocol in parallel to amplify completeness and soundness. However, this requires that the original private-coin verifier is called several times. Our main result can be stated as follows:
\begin{theorem}\label{thm:GoldwasserSipserUsingSampling}
	For any functions $c, s, \eps, \delta: 
	\mathbb{N}\rightarrow (0,1)$, and 
	$k, t, m, \ell: \mathbb{N} \rightarrow \mathbb{N}$, if $1/\eps, 1/\delta, k, t, m, \ell$ are time-constructible, then
	\begin{align*}
		&\IP\left(
			\begin{array}{l@{\hspace{0.4em}}l@{\hspace{0.4em}}l@{\hspace{0.4em}}}
					\textsf{rounds}& =    & k	\\
					\textsf{time}  & =    & t \\
					\textsf{msg size}  & =    & m \\
					\textsf{coins}  & =    & \ell \\
					\textsf{compl} & \geq & c \\
					\textsf{sound} & \leq & s \\
			\end{array}
		\right) \subseteq
		\AM\left(
			\begin{array}{l@{\hspace{0.4em}}l@{\hspace{0.4em}}l@{\hspace{0.4em}}}
					\textsf{rounds}& =    & 4k+3	\\
					\textsf{time}  & =    & t+ k\cdot \poly
							\left((m+\ell) \cdot \left(\frac{1}{\eps}\right)
																 ^{1/\delta}\right) 	\\
					\textsf{msg size}  & =    & \poly
							\left((m+\ell) \cdot \left(\frac{1}{\eps}\right)
																 ^{1/\delta}\right) \\
					\textsf{coins}  & =    & \poly
							\left((m+\ell) \cdot \left(\frac{1}{\eps}\right)
																 ^{1/\delta}\right) \\
					\textsf{compl} & \geq & c-2(k+1)\eps \\
					\textsf{sound} & \leq & (1+\eps + \delta)^{k+1} 
																	s +(k+1)\eps  \\
			\end{array}
		\right)
	\end{align*}
	Moreover, in the protocol that achieves this
	transformation, the public-coin verifier calls the 
	private-coin verifier exactly once.
\end{theorem} 

We state the following corollary that shows two interesting special cases for specific parameter choices. 
\begin{corollary}
	\label{cor:GoldwasserSipserUsingSampling}
	The following inclusions hold:
	\begin{enumerate}
		\item[(i)]
			For any functions $t,m, \ell: \mathbb{N}\rightarrow \mathbb{N}$, polynomial $k(n)$ and inverse polynomial $\gamma(n)<1/5$, if $t,m,\ell,k$ and $1/\gamma$ are time-constructible, we have
\end{enumerate}
	\begin{align*}
		&\IP\left(
			\begin{array}{lll}
					\textsf{rounds}& =    & k	\\
					\textsf{time}  & =    & t \\
					\textsf{msg size}  & =    & m \\
					\textsf{coins}  & =    & \ell \\
					\textsf{compl} & \geq & 2/3 + \gamma \\
					\textsf{sound} & \leq & 2^{-(k+5)} \\
			\end{array}
		\right) \subseteq
		\AM\left(
			\begin{array}{lll}
					\textsf{rounds}& =    & 4k+3	\\
					\textsf{time}  & =    & t+ 
													\poly((m+\ell)\frac{k}{\gamma})	\\
					\textsf{msg size}  & =    & \poly((m+\ell)\frac{k}{\gamma}) \\
					\textsf{coins}  & =    & \poly((m+\ell)\frac{k}{\gamma}) \\
					\textsf{compl} & \geq & 2/3 \\
					\textsf{sound} & \leq & 1/3 \\
			\end{array}
		\right)
	\end{align*}	

\begin{enumerate}
		\item[(ii)] For any functions $k,t,m,\ell:\mathbb{N}\rightarrow\mathbb{N}$, and $\gamma, \nu: \mathbb{N}\rightarrow (0,1)$, $\gamma \leq \nu$, if $k,t,m,\ell,1/\gamma, 1/\nu$ are time-constructible, then
\end{enumerate}
		\begin{align*}
		&\IP\left(
			\begin{array}{l@{\hspace{0.3em}}l@{\hspace{0.3em}}l@{\hspace{0.3em}}}
					\textsf{rounds}& =    & k	\\
					\textsf{time}  & =    & t \\
					\textsf{msg size}  & =    & m \\
					\textsf{coins}  & =    & \ell \\
					\textsf{compl} & \geq & 2/3 + \gamma \\
					\textsf{sound} & \leq & 1/3-\nu \\
			\end{array}
		\right) \subseteq
		\AM\left(
			\begin{array}{l@{\hspace{0.3em}}l@{\hspace{0.3em}}l@{\hspace{0.3em}}}
					\textsf{rounds}& =    & 4k+3	\\
					\textsf{time}  & =    & t+ 
							\poly\left((m+\ell)\cdot (\frac{k}{\gamma})
																						^{k/\nu}\right)\\
					\textsf{msg size}  & =    & \poly\left((m+\ell)\cdot (\frac{k}{\gamma})
																						^{k/\nu}\right) \\
					\textsf{coins}  & =    & \poly\left((m+\ell)\cdot (\frac{k}{\gamma})
																						^{k/\nu}\right) \\
					\textsf{compl} & \geq & 2/3 \\
					\textsf{sound} & \leq & 1/3 \\
			\end{array}
		\right)
	\end{align*}	
	Moreover, in the protocols that achieve the
	transformations in (i) and (ii), the public-coin verifier
	calls the private coin verifier exactly once.
\end{corollary}
The proof can be found in Section~\ref{sec:ProofOfSamplingCorollaries}. 
Note that part (i) implies $\IP(\textsf{rounds}=k) \subseteq \AM(\textsf{rounds}=4k+3)$ (where $\textsf{time}=\poly(n), \textsf{compl}\geq 2/3, \textsf{sound}\leq 1/3$), as the error probabilities of the $\IP$ protocol can be decreased by repeating it in parallel. As mentioned above, due to the repetition, the private coin verifier needs to be called multiple times in the resulting protocol. 

Next, we would like to compare our result to the \cite{GS86} transformation. For this comparison, the theorem below states what their transformation achieves \textit{after} repeating the private-coin protocol in parallel, i.e.~we again consider the setting where the private-coin verifier is called exactly once by the public-coin verifier. 
\begin{theorem}[\cite{GS86}]
	\label{thm:GS86}
	For any time-constructible polynomials $k(n)$, $t(n)$, 
	$m(n)$, and $\ell(n)$ we have
	\begin{align*}
		&\IP\left(
			\begin{array}{l@{\hspace{0.5em}}l@{\hspace{0.5em}}l@{\hspace{0.5em}}}
					\textsf{rounds}& =    & k	\\
					\textsf{time}  & =    & t \\
					\textsf{msg size}  & =    & m \\
					\textsf{coins}  & =    & \ell \\
					\textsf{compl} & \geq & 1-\ell^{-12k^2} \\
					\textsf{sound} & \leq & \ell^{-12k^2} \\
			\end{array}
		\right) \subseteq
		\AM\left(
			\begin{array}{l@{\hspace{0.5em}}l@{\hspace{0.5em}}l@{\hspace{0.5em}}}
					\textsf{rounds}& =    & k+2	\\
					\textsf{time}  & =    & t+\poly((m+\ell)k)\\
					\textsf{msg size}  & =    & \poly(m+\ell) \\
					\textsf{coins}  & =    & \poly(m+\ell) \\
					\textsf{compl} & \geq & 2/3 \\
					\textsf{sound} & \leq & 1/3 \\
			\end{array}
		\right)
	\end{align*}	
	Moreover, in the protocol that achieves this
	transformation, the public-coin verifier calls the private coin verifier exactly once.
\end{theorem}
By first repeating the $\IP$ protocol in parallel, this implies $\IP(\textsf{rounds}=k) \subseteq \AM(\textsf{rounds}=k+2)$. 
Comparing this theorem to our Corollary~\ref{cor:GoldwasserSipserUsingSampling} (i), we see that our result only 
loses a polynomially small fraction $\gamma$ in completeness, and
both losses in soundness and completeness are independent of $\ell$. 
If we apply Corollary~\ref{cor:GoldwasserSipserUsingSampling} (ii) for constant $k$, polynomial $t(n)$, any $\gamma$ that is inverse polynomial in $n$, and any constant $\nu > 0$, we obtain a verifier that runs in polynomial time. This is stronger than \cite{GS86} in all parameters except the number of rounds.

If there exists a constant-round sampling protocol that achieves the guarantees of our theorem, but the verifier runs in time $\poly(\frac{n}{\eps \delta})$, then we would get the transformation
\begin{align*}
		&\IP\left(
			\begin{array}{lll}
					\textsf{rounds}& =    & k	\\
					\textsf{time}  & =    & t \\
					\textsf{compl} & \geq & 2/3 + \gamma \\
					\textsf{sound} & \leq & 1/3 - \gamma \\
			\end{array}
		\right) \subseteq
		\AM\left(
			\begin{array}{lll}
					\textsf{rounds}& =    & \Theta(k)	\\
					\textsf{time}  & =    & 
											\poly(\frac{k\cdot t}{\gamma}) \\
					\textsf{compl} & \geq & 2/3 \\
					\textsf{sound} & \leq & 1/3 \\
			\end{array}
		\right)
\end{align*}
for any polynomially small $\gamma$ and polynomial $k$ (where the private-coin verifier is called exactly once). It is an interesting open problem if this can be achieved. 

\subsection{The Protocol}
\label{sec:PubCoinTheProtocol}
Let $(V,P)$ be the $\IP$ (i.e.~private-coin) protocol for some language $L$, as given in the theorem. We assume that on input $x$, $V$ chooses randomness $r \in \{0,1\}^{\ell(|x|)}$. 
For a fixed input $x \in \{0,1\}^n$ we 
let $\ell = \ell(n)$, and $M_0, \ldots, M_{k-1}$, and 
$A_0, \ldots, A_{k-1}$ be the random variables over the choice of $r$ that correspond to $V$'s messages $m_0, \ldots m_{k-1}$ and $P$'s answers $a_0,\ldots, a_{k-1}$ in the protocol execution $(V,P)(x)$. Furthermore, we let $\Gamma_i := (M_0, A_0, \ldots, M_{i},
A_{i})$ be the random variable over the entire communication.
We now describe the protocol $(V',P')$ that achieves the transformation described by Theorem~\ref{thm:GoldwasserSipserUsingSampling}.
The protocol will use the sampling protocol of 
Theorem~\ref{thm:main} several times, 
always using parameters $\delta$ and $\eps$.
$(V',P')(x)$ is defined as follows: 

\begin{itemize}
	\item[] \textbf{For} $i=0,\ldots, k-1$ \textbf{do}
	\begin{itemize}
		\item[] \textbf{Prover and Verifier:}
						Use the sampling protocol to sample $(m_i,p_i)$.
						
						If the prover is honest, it honestly executes
						the sampling protocol for the distribution $\Pd$
						defined by
						$$\Pd(m_i) := 
						\Pr_r[M_i = m_i|
							\Gamma_{i-1}
						= (m_0, a_0, \ldots, m_{i-1}, a_{i-1})].$$
		\item[] \textbf{Prover:}
						Send $a_i$ to the verifier.
						
						If the prover is honest, it sends 
						$a_i:=P(x,i,m_0, \ldots, m_i)$.\footnote{Recall that $P(x,i,m_0, \ldots, m_i)$ is the prover's answer in round $i$ given input $x$ and the previous verifier messages $m_i$.}
	\end{itemize}
	\item[] \textbf{Prover and Verifier:}
					Use the sampling protocol to sample $(r^*,p_k)$.
						
					If the prover is honest, it honestly executes
					the sampling protocol for the distribution $\Pd$
					defined by
					$$\Pd(r^*) := 
						\Pr_r[r^*=r|
							\Gamma_{k-1}
						= (m_0, a_0, \ldots, m_{k-1}, a_{k-1})].$$
	\item[] \textbf{Verifier:}
					Accept if and only if all of the following 
					conditions hold:
					\begin{enumerate}[(a)]
						\item $V(x,k,r^*,m_0, a_0, \ldots, m_{k-1},
									 a_{k-1}) = \accept$
						\item $\prod_{i=0}^{k} p_i = \frac{1}{2^{\ell}}$
					\end{enumerate}
\end{itemize}

\subsection{Technical Overview}
\label{sec:PubCoinTechOverview}
Theorem~\ref{thm:GoldwasserSipserUsingSampling} follows by applying the completeness and soundness guarantees of the sampling protocol. Completeness is quite easy to see, since the completeness of the sampling protocol guarantees that the correct probabilities are provided, and the distribution over the $m_i$'s (or $r$'s in the last round) is statistically close to the true distribution. In the following, we sketch the soundness proof. 

To simplify the exposition we make two assumptions. The proof is not much harder when these assumptions are omitted, and we refer to the full proof in Section~\ref{sec:PubCoinAnalysis} for details. First, we assume that the event $\Bad$ never occurs (i.e.~it has probability $0$) in the sampling protocol. Second, we only prove soundness for deterministic provers whose answers $a_i$ are fixed given the previous $(m_0, p_0, a_0, \ldots, m_{i-1},p_{i-1},a_{i-1}, m_i, p_i)$.\footnote{
	We cannot assume this in general, as the prover may for example choose $a_i$ depending on a hash function that was chosen in some earlier execution of the sampling protocol. 
} 

So suppose $x\notin L$ and let $P^*$ be any prover that satisfies the above assumption. We denote the randomness of $V'$ by $r'$, let $M_i', P_i', A_i', R^*$ be the random variables over the choice of $r'$ corresponding to $m_i, p_i, a_i, r^*$ in the protocol $(V',P^*)(x)$, and define $\Gamma_i':=(M_0', P_0', A_0', \ldots, M_i', P_i', A_i')$. 

We first fix any $\gamma'_{k-1} = (m_0, p_0, a_0, \ldots, m_{k-1}, p_{k-1}, a_{k-1})$ that occurs with nonzero probability, and let $\gamma_{k-1}=(m_0, a_0, \ldots, m_{k-1}, a_{k-1})$. Now we bound the acceptance probability conditioned on $\gamma'_{k-1}$ as follows (we discuss the steps of this calculation below).
\begin{align}
  \Pr_{r'}&\left[
				(V'(r'),P^*)(x)=\accept | \Gamma_{k-1}'=\gamma_{k-1}'
				\right] \nonumber\\
	&= \sum_{r^*\in \{0,1\}^{\ell}}
		 \Pr_{r'}\left[
							(V'(r'),P^*)(x)=\accept 
							\land R^*=r^*
							| \Gamma_{k-1}'=\gamma_{k-1}'
					 \right] \nonumber\\
	&= \sum_{r^*}
		 \Pr_{r'}\Bigr[
							V(x,k,r^*,\gamma_{k-1})=\accept
							\land P_k=\frac{1}{2^{\ell}\prod_{i=0}^{k-1}p_i}
							\land R^*=r^* \nonumber\\
							&\qquad\qquad\qquad\qquad\qquad\qquad
							 \qquad\qquad\qquad\qquad\qquad\qquad
							\Big| \Gamma_{k-1}'=\gamma_{k-1}'
					 \Bigl] \nonumber
\end{align}
\begin{align}					
	&= \sum_{r^*}
		 \left[V(x,k,r^*,\gamma_{k-1})=\accept \right] \nonumber\\
		 &\qquad\cdot
		 \Pr_{r'}\left[
							P_k=\frac{1}{2^{\ell}\prod_{i=0}^{k-1}p_i}
							\land R^*=r^*
							\Big| \Gamma_{k-1}'=\gamma_{k-1}'
					 \right] \nonumber\\
	&\leq
		 \sum_{r^*}
		 \left[V(x,k,r^*,\gamma_{k-1})=\accept \right] 
		 \cdot
		 \frac{1+\eps+\delta}{2^{\ell}\prod_{i=0}^{k-1}p_i}
		\nonumber\\
	&=\frac{1+\eps+\delta}{\prod_{i=0}^{k-1}p_i}
		\cdot \Pr_{r \in \{0,1\}^{\ell}}
			\left[
				V(x,k,r,\gamma_{k-1})=\accept
			\right] \label{eqn:pubcoinsketch1}
\end{align}
We used the notation $[\cdot]$ for a Boolean expression, which is defined to be $1$ iff the expression is true. Furthermore, we applied the soundness guarantee of the sampling protocol to obtain the above inequality. 

Using this, we can proceed to bound the acceptance probability
conditioned on any fixed $\gamma'_{k-2} = (m_0, p_0, a_0, \ldots, m_{k-2}, p_{k-2}, a_{k-2})$ as follows (we again discuss the steps of this calculation below). 
\begin{align*}
  \Pr_{r'}&\left[
				(V'(r'),P^*)(x)=\accept | \Gamma_{k-2}'=\gamma_{k-2}'
				\right] \\
	&=\sum_{m_{k-1},p_{k-1}}
		\Pr_{r'}\bigl[
				(V'(r'),P^*)(x)=\accept \\
		&\qquad\qquad\qquad\qquad	\big| 
					(\Gamma_{k-2}',M_{k-1}', P_{k-1}')= 
					(\gamma_{k-2}',m_{k-1}, p_{k-1})
				\bigr]\\
		&\qquad \cdot
		\Pr_{r'}\left[
				(M_{k-1}', P_{k-1}')= 
					(m_{k-1}, p_{k-1})
				| \Gamma_{k-2}'= \gamma_{k-2}'
				\right]\\
	&=\sum_{m_{k-1},p_{k-1}}
		\Pr_{r'}\bigl[
				(V'(r'),P^*)(x)=\accept \\
		&\qquad\qquad\qquad \qquad	\big|
					(\Gamma_{k-1}')= 
					(\gamma_{k-2}',m_{k-1}, p_{k-1}, a_{k-1}^*)
				\bigr]\\
		&\qquad \cdot
		\Pr_{r'}\left[
				(M_{k-1}', P_{k-1}')= 
					(m_{k-1}, p_{k-1})
				| \Gamma_{k-2}'= \gamma_{k-2}'
				\right]\\
	&\leq 
		\sum_{m_{k-1},p_{k-1}}
		\frac{1+\eps+\delta}{\prod_{i=0}^{k-1}p_i}
		\cdot \Pr_{r}
			\left[
				V(x,k,r,\gamma_{k-2},m_{k-1},a_{k-1}^*)=\accept
			\right]\\
		&\qquad \cdot
		\Pr_{r'}\left[
				(M_{k-1}', P_{k-1}')= 
					(m_{k-1}, p_{k-1})
				| \Gamma_{k-2}'= \gamma_{k-2}'
				\right]\\
	&=\frac{1+\eps+\delta}{\prod_{i=0}^{k-2}p_i}
		\sum_{m_{k-1}}
			\Pr_{r}
			\left[
				V(x,k,r,\gamma_{k-2},m_{k-1},a_{k-1}^*)=\accept
			\right] \\
		&\qquad \cdot
		 \underbrace{
			 \sum_{p_{k-1}} \frac
				{
					\Pr_{r'}\left[
					(M_{k-1}', P_{k-1}')= 
						(m_{k-1}, p_{k-1})
					| \Gamma_{k-2}'= \gamma_{k-2}'
					\right]
				}
				{	p_{k-1} }	
			}_{
				\leq (1+\eps+\delta)
			} 
	\end{align*}
	\begin{align*}
		&\leq
			\frac{(1+\eps+\delta)^2}{\prod_{i=0}^{k-2}p_i}
			\sum_{m_{k-1}}
				\Pr_{r}
				\left[
					V(x,k,r,\gamma_{k-2},m_{k-1},a_{k-1}^*)=\accept
				\right]
\end{align*}
The second inequality above follows by our assumption that 
for fixed values $\gamma'_{k-2},m_{k-1},p_{k-1}$, the prover's answer $A_{k-1}$ is also fixed. We denote this value by $a_{k-1}^*$ and stress that formally this is a function of $(\gamma'_{k-2},m_{k-1},p_{k-1})$. We proceed to use this notation below, and in the given context the values $a_{i}^*$ must always be interpreted as a function of $(\gamma'_{i-1},m_{i},p_{i})$. This notation is somewhat sloppy, but we use it here to simplify the exposition, and treat the issue more formally in the full soundness analysis of Section~\ref{sec:PubCoinProofSoundnessOverview}. 

The first inequality above follows from Eqn.~\ref{eqn:pubcoinsketch1}, and the underbraced term is bounded using the soundness guarantee of the sampling protocol (noting that in each execution of the sampling protocol, the verifier uses independent random coins). 

Finally, repeating the above calculation conditioning on $\gamma'_{k-3}, \gamma'_{k-4}$ and so on yields that
\begin{align*}
	\Pr_{r'}&\left[
				(V'(r'),P^*)(x)=\accept 
				\right] \leq (1+\eps+\delta)^{k+1}\\
	& \qquad\qquad
		\cdot \sum_{m_0}\sum_{m_1}\ldots\sum_{m_{k-1}}
		\Pr_r\left[
			V(x,k,r,m_0, a_0^*, \ldots, m_{k-1}, a_{k-1}^*)
			=\accept
		\right]\\
	&\leq (1+\eps+\delta)^{k+1}s,
\end{align*}
where for the second inequality we used that $V$ only accepts if the transcript is consistent, and applied the soundness guarantee of the protocol $(V,P)$. 

\paragraph{Removing the simplifying assumptions.}
We briefly sketch how the two simplifying assumptions in the proof sketch can be waived. 

First, as the event $\Bad$ occurs with probability at most $\eps$ for each sampling protocol invocation, one simply conditions the analysis on no bad event occurring. This induces the additional error term of $(k+1)\eps$. 

Second, we show below that the proof above can be modified to deal with general provers who may send arbitrary answers $a_i$. This is quite intuitive, and our analysis shows that the best strategy for the prover is to send the $a_i$ that maximize the acceptance probability of the verifier $V$.

\subsection{Analysis of the Protocol}
\label{sec:PubCoinAnalysis}

\subsubsection{Proof of completeness}
	Suppose $x \in L$ and fix the honest prover $P'$. 
	For any $i \in (k-1)$ the following holds.
	For any $\gamma_i = (m_0, a_0, \ldots m_i, a_i)$, where 
	$a_i = P(x,i,m_0, \ldots, m_i)$, letting
	$\gamma_{i-1} = (m_0, a_0, \ldots, m_{i-1}, a_{i-1})$, we 
	have
	\begin{align}
		\Pr_r\left[\Gamma_{i}
		  	 =\gamma_i\right]
		&= \Pr_r\left[(M_i, A_i) = 
					  (m_i, a_i)
						|
						\Gamma_{i-1}
					 =\gamma_{i-1} \right]
		\cdot 
		\Pr_r\left[\Gamma_{i-1}
					 =\gamma_{i-1}\right] \nonumber \\
		&= \Pr_r\left[M_i = 
					  m_i
						|
						\Gamma_{i-1}
					 =\gamma_{i-1}\right]
		\cdot 
		\Pr_r\left[\Gamma_{i-1}
					 =\gamma_{i-1}\right], \label{eqnp:50}
	\end{align}
	where the first equality holds by the definition of conditional probabilities, and the second inequality holds because the prover is deterministic and $a_i = P(x,i,m_0, \ldots, m_i)$. 
	
	Furthermore, since any fixed $r$ uniquely determines 
	$\Gamma_{k-1}$, we find that for any $r^*$ and any
	$\gamma_{k-1} = (m_0, a_0, \ldots m_{k-1}, a_{k-1})$ where 
	$(x,r^*, \gamma_{k-1})$ is consistent for $V$	we have
	\begin{align}
		\Pr_r\left[r^*=r | \Gamma_{k-1}
		  	 =\gamma_{k-1}\right]
		\cdot
		  \Pr_r\left[\Gamma_{k-1}
							  =\gamma_{k-1}\right]
		= 
			\Pr_r\left[r^*=r\right] = \frac{1}{2^{\ell}}.\label{eqnp:51}
	\end{align}
	The completeness of the sampling protocol implies that 
	if some output $m$ is generated, then it is generated along
	with $\Pd(m)$.
	Together with (\ref{eqnp:50}) and (\ref{eqnp:51}),	this
	implies that in case the verifier does not
	reject in the sampling protocols, (b) holds.
	
	The next claim states that the output distribution given by the
	sampling protocol (when interacting with the honest prover for some distribution $\Pd$) is close
	to the distribution $\Pd$. 
	\begin{claim}
		\label{claim:SmallStatDist}
		Fix a distribution $\Pd$ over $\{0,1\}^n$, and 
		$\delta, \eps \in (0,1)$. 		
		Let $\Pd'$ be the following distribution over 
		$\{0,1\}^n \cup \{\bot\}$: 
		Run the sampling protocol of Theorem~\ref{thm:main} for
		$\Pd$ with parameters $\delta, \eps$ and the honest 
		prover. If the verifier rejects, return $\bot$. 
		Otherwise, let $(x,p)$ be the verifier's output and 
		return $x$. Then the statistical distance between 
		$\Pd$ and $\Pd'$ is at most $2\eps$. 
	\end{claim}
	\begin{proof}
		Let $\cM$ be the set given by the completeness 
		guarantee of the sampling protocol. 
		By completeness, for any $x\in \cM$ we have
		$\Pd'(x) \in (1\pm\eps)\Pd(x)$. Furthermore, we have
		$\Pd(\overline{\cM}) \leq	\eps$. This implies the claim.
	\end{proof}
	
	Thus, by coupling, we may assume that in each invocation of
	the sampling protocol, the bitstring	is sampled according
	to the distribution $\Pd$, and this
	fails with probability at most $2\eps$. Thus, with
	probability $1-(k+1)2\eps$ we simulate an	execution
	of the protocol $(V,P)$, which accepts with probability
	at least $c$.

\subsubsection{Proof of soundness: overview}
	\label{sec:PubCoinProofSoundnessOverview}
	Suppose $x\notin L$ and let $P^*$ be any deterministic prover.\footnote{
	By a standard argument, we may assume that the prover is deterministic: for any probabilistic prover, we consider the deterministic prover that maximizes the verifier's acceptance probability. 
	}
	We denote the randomness of $V'$ by $r'$. 
	For $i \in (k)$ let $\Bad_i$ be the event that occurs 
	if in the $i$'th execution of the sampling protocol
	the event $\Bad$ as described in Theorem~\ref{thm:main}
	occurs. 
	Furthermore, let $\Bad' = \Bad_0 \lor \ldots 
	\lor \Bad_{k}$.
	As $V'$ uses independent 
	random coins for each execution of the sampling protocol, 
	and by Theorem~\ref{thm:main} each event $\Bad_i$ occurs 
	with probability at most $\eps$, we find
	that $\Pr_{r'}[(V'(r'),P^*)(x)=\accept]$ equals
	\begin{align*}
	  &\Pr_{r'}[(V'(r'),P^*)(x)=\accept | \lnot \Bad']
		   \Pr_{r'}[\lnot \Bad'] \\
		&\qquad +
		   \Pr_{r'}[(V'(r'),P^*)(x)=\accept | \Bad']
			 \Pr_{r'}[\Bad'] \\
		\leq &\Pr_{r'}[(V'(r'),P^*)(x)=\accept | \lnot \Bad']
		 +
			 \Pr_{r'}[\Bad'] \\
		\leq
		&\Pr_{r'}[(V'(r'),P^*)(x)=\accept | \lnot \Bad'] + 
		(k+1)\eps.		
	\end{align*}

	We proceed to bound the acceptance probability conditioned
	on $\lnot \Bad'$. Let 
	$M'_i, P'_i, A'_i, R^*$ be the random variables over the 
	choice of $r'$ corresponding to
	$m_i, p_i, a_i, r^*$ in the protocol	$(V', P^*)(x)$, and
	additionally define $\Gamma'_{i} := (M'_0, P'_0, A'_0,
	\ldots, M'_i, P'_i, A'_i)$.

	\begin{claim} 
	\label{claim:LastSamplingStep}
	Fix any $\gamma'_{k-1} = (m_0, p_0, a_0, \ldots, 
	m_{k-1}, p_{k-1}, a_{k-1})$ that satisfies
		$
			\Pr_{r'}\left[\Gamma'_{k-1} 
										= \gamma'_{k-1} \land \lnot \Bad'
							\right] > 0
		$, and let 
		$\gamma_{k-1} = (m_0, a_0, \ldots, m_{k-1},a_{k-1})$. 
		Then we have
		\begin{align}
			\Pr_{r'}&\left[(V'(r'),P^*)(x)=\accept
						 |\Gamma'_{k-1}=\gamma'_{k-1} 
							\land \lnot \Bad'\right] \label{eqn:GS1} \\
							&\leq 
				\frac{1+\eps+\delta}
						 {\prod_{i=0}^{k-1}p_i}
				\Pr_{r\leftarrow \{0,1\}^{\ell}}\left[
					V(x,k,r,\gamma_{k-1}) = \accept
					\right].\nonumber
		\end{align}
	\end{claim}
	As the proof is not difficult, we only give a proof sketch. The formal proof can be found in Section~\ref{sec:GSSoundnessDetails}.	
	\begin{proof}[Proof	(Sketch)]
		By definition of the protocol, the verifier $V'$ 
		accepts if 
		and only if $V(x,k,R^*,\gamma_{k-1})=\accept$
		and $P_k = \frac{1}{2^{\ell}\prod_{i=0}^{k-1}p_i}$. 
		The soundness condition (ii) of the sampling protocol
		guarantees that for any $r^*$ the probability of 
		$R^*=r^*\land \P_k=\frac{1}{2^{\ell}\prod_{i=0}^{k-1}p_i}$
		(under the given conditions) is bounded by 
		$\frac{1+\eps+\delta}{2^{\ell}\prod_{i=0}^{k-1}p_i}$. 
		As $V$ chooses $r$ uniformly in $\{0,1\}^{\ell}$, 
		this gives the claim.
	\end{proof}
		
	\begin{claim}
		\label{claim:SamplingInduction}
		Fix any $i \in (k)$ and 
		$\gamma'_{i-1} = (m_0,p_0,a_0, \ldots, m_{i-1},p_{i-1},
		a_{i-1})$ with
		$\Pr_{r'}[\Gamma'_{i-1} = \gamma'_{i-1} \land \lnot 
		\Bad']>0$, and
		let $\gamma_{i-1} = (m_0, a_0, \ldots, m_{i-1}, a_{i-1})$. 
		Then we have 
		\begin{align*}	
			&\Pr_{r'}\left[
				(V'(r'),P^*)(x)=\accept
						 |\Gamma'_{i-1}=\gamma'_{i-1} \land \lnot \Bad'
			\right] \\
			&\leq
			\frac{(1+\eps+\delta)^{k-i+1}}
						 {\prod_{j=0}^{i-1}p_j}
			\sum_{m_{i}}\max_{a_{i}}
				\Bigl( \ldots \\
			&\qquad 
			\sum_{m_{k-1}}
				\max_{a_{k-1}}
			\Bigl(
					\Pr_{r\leftarrow \{0,1\}^m}\left[
					V(x,k,r,\gamma_{i-1},m_{i},a_i, \ldots, 
						m_{k-1},a_{k-1}) = \accept 
					\right]
			 \Bigr)
			 \Bigr).
		\end{align*}
	\end{claim}
	In the proof we only omit one step, which can be found 
	in Section~\ref{sec:GSSoundnessDetails}, 
	Claim~\ref{claim:GSIndHypoApplication}. 
	\begin{proof}[Proof]
		The proof is by induction for decreasing $i$. The base case
		for $i=k$ is given by Claim~\ref{claim:LastSamplingStep}.
		For the induction step, we assume the claim holds for 
		$i+1$, and our goal is to prove	it for $i$. 
		
		We note that
		$\Pr_{r'}\left[
				(V'(r'),P^*)(x)=\accept
						 |\Gamma'_{i-1}=\gamma'_{i-1} \land \lnot \Bad'
			\right]
		$ equals
		\begin{align}
			\sum_{p_{i},m_{i}}
				&\Pr_{r'}\left[
					(V'(r'),P^*)(x)=\accept
					| (\Gamma'_{i-1},M'_{i},P'_{i})
					 =(\gamma'_{i-1},m_{i},p_{i}) \land \lnot \Bad'
				\right] \label{eqnp:52} \\
				&\qquad \qquad \cdot
				\Pr_{r'}\left[
					(M'_{i},P'_{i}) = (m_{i},p_{i})
					| \Gamma'_{i-1}
					 =\gamma'_{i-1} \land \lnot \Bad'
				\right]. \nonumber
		\end{align}
	Now the probability in (\ref{eqnp:52}) can be bounded
	using the induction hypothesis. (We do this formally 
	in Claim~\ref{claim:GSIndHypoApplication}.) This gives that
	the above sum equals
	\begin{align*}
			&\sum_{p_{i},m_{i}}
				\frac{(1+\eps+\delta)^{k-i}}
						 {\prod_{j=0}^{i}p_j}
				\max_{a_i}\Bigl(\sum_{m_{i+1}}\max_{a_{i+1}}
							\Bigl( \ldots \\
						&\qquad 
						\sum_{m_{k-1}}
							\max_{a_{k-1}}
						\Bigl(
								\Pr_{r\leftarrow \{0,1\}^{\ell}}\left[
								V(x,k,r,\gamma_{i-1},
								m_{i},a_{i}, \ldots, 
									m_{k-1},a_{k-1}) = \accept 
								\right]
			 \Bigr)
			 \Bigr) \Bigr)\\
				&\qquad \cdot
				\Pr_{r'}\left[
					(M'_{i},P'_{i}) = (m_{i},p_{i})
					| \Gamma'_{i-1}
					 =\gamma'_{i-1} \land \lnot \Bad'
				\right] \\
			&= \frac{(1+\eps+\delta)^{k-i}}
						 {\prod_{j=0}^{i-1}p_j}
				\sum_{m_{i}}				
				\max_{a_i}\Bigl(\sum_{m_{i+1}}\max_{a_{i+1}}
							\Bigl( \ldots \\
						&\qquad 
						\sum_{m_{k-1}}
							\max_{a_{k-1}}
						\Bigl(
								\Pr_{r\leftarrow \{0,1\}^{\ell}}\left[
								V(x,k,r,\gamma_{i-1},
								m_{i},a_{i}, \ldots, 
									m_{k-1},a_{k-1}) = \accept 
								\right]
			 \Bigr)
			 \Bigr) \Bigr)\\
				&\qquad \qquad 
				\sum_{p_i}\frac{				
				\Pr_{r'}\left[
					(M'_{i},P'_{i}) = (m_{i},p_{i})
					| \Gamma'_{i-1}
					 =\gamma'_{i-1} \land \lnot \Bad'
				\right]}{p_i} 
	\end{align*}
	\begin{align*}
			&\leq
				\frac{(1+\eps+\delta)^{k-i+1}}
						 {\prod_{j=0}^{i-1}p_j}
				\sum_{m_{i}}				
				\max_{a_i}\Bigl(\sum_{m_{i+1}}\max_{a_{i+1}}
							\Bigl( \ldots \\
						&\qquad 
						\sum_{m_{k-1}}
							\max_{a_{k-1}}
						\Bigl(
								\Pr_{r\leftarrow \{0,1\}^{\ell}}\left[
								V(x,k,r,\gamma_{i-1},
								m_{i},a_{i}, \ldots, 
									m_{k-1},a_{k-1}) = \accept 
								\right]
			 \Bigr)
			 \Bigr) \Bigr),
		\end{align*}
	where the inequality follows from the soundness
	condition (ii) of the sampling protocol 
	(Theorem~\ref{thm:main}).
	\end{proof}

	Finally, applying Claim~\ref{claim:SamplingInduction} for 
	$i=0$ gives that
	\begin{align*}	
			&\Pr_{r'}\left[
				(V'(r'),P^*)(x)=\accept|\lnot \Bad'
			\right] \leq	
			(1+\eps+\delta)^{k+1}
			\sum_{m_{0}}\max_{a_{0}}
				\Bigl( \ldots \\
			&\qquad \sum_{m_{k-1}}
				\max_{a_{k-1}}
			\Bigl(
					\Pr_{r\leftarrow \{0,1\}^{\ell}}\left[
					V(x,k,r,m_{0},a_0, \ldots, 
						m_{k-1},a_{k-1}) = \accept 
					\right]
			 \Bigr)
			 \Bigr).
		\end{align*}
		This finishes the proof, since
		the soundness of the protocol $(V,P)$ implies that
		$$\sum_{m_{0}}\max_{a_{0}}
				\Bigl( \ldots 
			\sum_{m_{k-1}}
				\max_{a_{k-1}}
			\Bigl(
					\Pr_{r\leftarrow \{0,1\}^{\ell}}\left[
					V(x,k,r,m_{0},a_0, \ldots, 
						m_{k-1},a_{k-1}) = \accept 
					\right]
			 \Bigr)
			 \Bigr)$$
is at most $s$. This inequality holds because according to Definition~\ref{def:InteractiveProtocol}, the verifier always rejects in case the protocol transcript is not consistent.

\subsubsection{Proof of soundness: the details}
\label{sec:GSSoundnessDetails}
	\begin{proof}[Proof of Claim~\ref{claim:LastSamplingStep}]
		The probability in (\ref{eqn:GS1}) equals
		\begin{align*}
			&\sum_{r^* \in \{0,1\}^{\ell}} 
				\Pr_{r'}\left[(V'(r'),P^*)(x)=\accept
							\land R^* = r^*
						 |\Gamma'_{k-1} = \gamma'_{k-1}
							\land \lnot \Bad'\right]\\
			&=\sum_{r^*} \Pr_{r'}\Bigl[
							V(x,k,r^*,\gamma_{k-1}) = \accept
			\land
							P_k = \frac{1}{2^{\ell} \cdot \prod_{i=0}^{k-1}p_i}
							\land R^* = r^*\\
						 &\qquad\qquad\qquad\qquad\qquad\qquad\qquad
							\qquad\qquad\qquad
						 \Big|\Gamma'_{k-1} = \gamma'_{k-1}
							\land \lnot \Bad'\Bigr]\\
			&= \sum_{r^*} \Pr_{r'}\Bigl[
						 P_k = \frac{1}{2^{\ell} \cdot \prod_{i=0}^{k-1}p_i}
							\land R^* = r^*
						 \Big|\Gamma'_{k-1} = \gamma'_{k-1}
							\land \lnot \Bad'\Bigr]\\
			&\quad \cdot \Pr_{r'}\Bigl[
						 V(x,k,r^*,\gamma_{k-1}) = \accept\\
						 &\qquad\qquad\qquad
						 \Big|\Gamma'_{k-1} = \gamma'_{k-1} 
							\land \lnot \Bad'
							\land P_k = 
									\frac{1}{2^{\ell} \cdot \prod_{i=0}^{k-1}p_i}
							\land R^* = r^*
							\Bigr]
	\end{align*}
	\begin{align*}
			&\leq\frac{1+\eps+\delta}
						{2^{\ell} \cdot \prod_{i=0}^{k-1}p_i} 
					\sum_{r^*} 										
					\left[ 
						V(x,k,r^*,\gamma_{k-1}) = \accept
					\right] \\
			&=\frac{1+\eps+\delta}
				{\prod_{i=0}^{k-1}p_i} 
				\Pr_{r\leftarrow \{0,1\}^{\ell}}\left[	
				V(x,k,r,\gamma_{k-1}) = \accept 
				\right],
		\end{align*}
		where the inequality is implied by the soundness
		condition (ii) of the sampling protocol 
		(Theorem~\ref{thm:main}).
	\end{proof}
	
	\begin{claim}
		\label{claim:GSIndHypoApplication}
		Assuming Claim~\ref{claim:SamplingInduction} holds for 
		$i+1$, we have for any $(p_i,m_i)$ that
		\begin{align*}
			&\Pr_{r'}\left[
					(V'(r'),P^*)(x)=\accept
					| (\Gamma'_{i-1},M'_{i},P'_{i})
					 =(\gamma'_{i-1},m_{i},p_{i}) \land \lnot \Bad'
				\right] \\
			&\qquad \leq
			\frac{(1+\eps+\delta)^{k-i}}
						 {\prod_{j=0}^{i}p_j}
						\max_{a_i}\Bigl(\sum_{m_{i+1}}\max_{a_{i+1}}
							\Bigl( \ldots \\
						&\qquad 
						\sum_{m_{k-1}}
							\max_{a_{k-1}}
						\Bigl(
								\Pr_{r\leftarrow \{0,1\}^{\ell}}\left[
								V(x,k,r,\gamma_{i-1},
								m_{i},a_{i}, \ldots, 
									m_{k-1},a_{k-1}) = \accept 
								\right]
			 \Bigr)
			 \Bigr) \Bigr).
		\end{align*}
	\end{claim}
	
	\begin{proof}[Proof]
		We find
		\begin{align*}
			&\Pr_{r'}\left[
					(V'(r'),P^*)(x)=\accept
					| (\Gamma'_{i-1},M'_{i},P'_{i})
					 =(\gamma'_{i-1},m_{i},p_{i}) \land \lnot \Bad'
				\right] \\
			&= \sum_{a_{i}}
				\Pr_{r'}\bigl[
					(V'(r'),P^*)(x)=\accept \\
					&\qquad\qquad\qquad
					\big| (\Gamma'_{i-1},M'_{i},P'_{i},A'_{i})
					 =(\gamma'_{i-1},m_{i},p_{i},a_{i}) 
					\land \lnot \Bad'
				\bigr] \\
		  &\qquad \cdot
				\Pr_{r'}\left[
					A'_{i} = a_{i}
					| (\Gamma'_{i-1},M'_{i},P'_{i})
					 =(\gamma'_{i-1},m_{i},p_{i}) 
					\land \lnot \Bad'
				\right] \\
			&\leq
				\sum_{a_{i}}
				\frac{(1+\eps+\delta)^{k-i}}
						 {\prod_{j=0}^{i}p_j}
						\sum_{m_{i+1}}\max_{a_{i+1}}
							\Bigl( \ldots \\
						&\qquad 
						\sum_{m_{k-1}}
							\max_{a_{k-1}}
						\Bigl(
								\Pr_{r\leftarrow \{0,1\}^{\ell}}\left[
								V(x,k,r,\gamma_{i-1},
								m_{i},a_{i}, \ldots, 
									m_{k-1},a_{k-1}) = \accept 
								\right]
			 \Bigr)
			 \Bigr) \\
				&\qquad \cdot
				\Pr_{r'}\left[
					A'_{i} = a_{i}
					| (\Gamma'_{i-1},M'_{i},P'_{i})
					 =(\gamma'_{i-1},m_{i},p_{i}) 
					\land \lnot \Bad'
				\right]\\
			&\leq
				\frac{(1+\eps+\delta)^{k-i}}
						 {\prod_{j=0}^{i}p_j}
						\max_{a_i}\Bigl(\sum_{m_{i+1}}\max_{a_{i+1}}
							\Bigl( \ldots \\
						&\qquad 
						\sum_{m_{k-1}}
							\max_{a_{k-1}}
						\Bigl(
								\Pr_{r\leftarrow \{0,1\}^{\ell}}\left[
								V(x,k,r,\gamma_{i-1},
								m_{i},a_{i}, \ldots, 
									m_{k-1},a_{k-1}) = \accept 
								\right]
			 \Bigr)
			 \Bigr) \Bigr),
		\end{align*}
		where the first inequality follows from the assumption
		in the claim.
	\end{proof}

\subsection{Our Transformation for Specific Parameter Choices}
\label{sec:ProofOfSamplingCorollaries}
\begin{proof}[Proof of Corollary~\ref{cor:GoldwasserSipserUsingSampling}]
	Part (i) follows if we set $\eps := \frac{\gamma}{2(k+1)}$, and $\delta := 1/2$.
	
	We proceed to prove (ii). 
	First, Theorem~\ref{thm:GoldwasserSipserUsingSampling} gives for any $\delta$ and $\eps\leq\delta$ that
	\begin{align}
	&\IP\left(
			\begin{array}{lll}
					\textsf{rounds}& =    & k	 \\ 
					\textsf{time}  & =    & t 	 \\
					\textsf{msg size}  & =    & m \\
					\textsf{coins}  & =    & \ell \\
					\textsf{compl} & \geq & 2/3+2(k+1)\eps   \\ 
					\textsf{sound} & \leq & 1/3(1-3(k+1)\eps)2^{-4\delta(k+1)}   \\ 
			\end{array}
		\right) \nonumber	\\
		&\qquad \subseteq
		\AM\left(
			\begin{array}{lll}
					\textsf{rounds}& =    & 4k+3	 \\ 
					\textsf{time}  & =    & t+ k\cdot\poly((m+\ell)(1/\eps)^{1/\delta}) 	 \\
					\textsf{msg size}  & =    & \poly((m+\ell)(1/\eps)^{1/\delta}) \\
					\textsf{coins}  & =    & \poly((m+\ell)(1/\eps)^{1/\delta}) \\
					\textsf{compl} & \geq & 2/3 \\ 
					\textsf{sound} & \leq & 1/3 \\ 
			\end{array}
		\right).\label{eqn:GSCor1}
	\end{align}
	To see this, we apply the theorem for $c:=2/3+2(k+1)\eps$ and $s:=1/3(1-3(k+1)\eps)2^{-4\delta(k+1)}$
	and note that since $(1+2\delta) \leq 2^{4\delta}$ we have
	$$ s= 1/3(1-3(k+1)\eps)2^{-4\delta(k+1)} \leq \frac{1/3(1-3(k+1)\eps)}{(1+2\delta)^{k+1}}
			 \leq 
			 \frac{1/3(1-3(k+1)\eps)}{(1+\eps + \delta)^{k+1}}.
	$$ 
	
	Now given $\nu$ and $\gamma$, we set
	$\delta:=\frac{\nu}{4(k+1)}$, and $\eps:=\frac{\gamma}{4(k+1)}$. Since $\gamma \leq \nu$, we have
	$\eps \leq \delta$ and we can apply (\ref{eqn:GSCor1}). We find
	\begin{align*}
		&1/3(1-3(k+1)\eps)2^{-4\delta(k+1)} =
		1/3(1-3(k+1)\eps)2^{-\nu}\\
		&\qquad \geq 
		1/3(1-3(k+1)\eps)(1-2\nu) 
		\geq
		1/3(1-3(k+1)\eps-2\nu)\\
		&\qquad = 1/3-(k+1)\eps - 2\nu/3
		\geq 1/3- \nu,
	\end{align*}
	where we used $\eps \leq \delta$ for the last inequality. Furthermore, we have
	$
		2/3 + 2(k+1)\eps = 2/3 + \gamma/2 \leq 2/3 + \gamma
	$.
	These two inequalities imply the claim.
\end{proof}

\section*{Acknowledgements}
We would like to thank Jan Hązła for the discussions about 
		the soundness condition of our protocol. In particular, 
		he pointed out that the condition can formally be seen as
		an upper bound that holds on average, which helped improve
		the exposition of this result a lot.

\bibliographystyle{alpha}
\bibliography{bibliography}

\end{document}